\documentclass[%
reprint,
 amsmath,amssymb,
 aps,
floatfix,
]{revtex4-2}

\usepackage{graphicx}
\usepackage{dcolumn}
\usepackage{bm}
\usepackage{hyperref}
\hypersetup{
    colorlinks=true,
    linkcolor=blue,
    filecolor=magenta,
    urlcolor=blue,
    citecolor=purple,
}

\usepackage{dsfont}
\usepackage{orcidlink}
\usepackage{sidecap}
\usepackage{fancyhdr}
\usepackage{fancyref}
\usepackage{xspace}
\usepackage[normalem]{ulem}
\usepackage{enumitem}
\usepackage{times}
\usepackage{booktabs}
\usepackage{tikz}
\usepackage[dvipsnames]{xcolor}
\usepackage{comment}
\usepackage{wrapfig}
\usepackage{amsthm}
\usepackage{algpseudocode}
\usepackage{caption}

\newtheorem{theorem}{Theorem}[section]
\newtheorem{lemma}[theorem]{Lemma}

\begin{document}

\preprint{APS/123-QED}

\title{Two Point Correlation Function Estimation with Contaminated Data}

\author{Arya~Farahi\orcidlink{0000-0003-0777-4618}}
 \email{arya.farahi@austin.utexas.edu}
\affiliation{%
 Department of Statistics and Data Sciences, University of Texas at Austin, Austin, Texas 78712, USA\\
 The NSF-Simons AI Institute for Cosmic Origins, USA 
}%

\date{\today}

\begin{abstract}
The two-point correlation function (2PCF) is a cornerstone of precision cosmology, yet its estimation from imaging surveys is vulnerable to contamination and incompleteness arising from imperfect target selection and pipeline-level inclusion decisions. In practice, the scientific target is a physically defined population (e.g., galaxies in a redshift or luminosity range), while the working catalog is constructed from noisy measurements and selection cuts, leading to mismatches between true and observed inclusion. These errors are rarely spatially uniform: they correlate with survey depth, observing conditions, and foreground structure, and can imprint spurious large-scale power or suppress the true clustering signal. High-resolution, high-SNR spectroscopic samples provide gold-standard inclusion in the target population, but are typically available for only a small subset of objects. We introduce a prediction-powered Landy–Szalay (PP–LS) estimator that combines noisy inclusion labels over the full catalog with exact labels on a small spectroscopic subset, while preserving the standard random-catalog normalization that corrects for survey geometry and selection. PP–LS debiases pair counts through residual-based, design-weighted correction terms computed only on the labeled subset, requiring no probability calibration, no known misclassification rates or class priors, no explicit spatial model of contamination, nor forward modeling of the systematics. Under simple random sampling of the labeled subset, we establish recovery of the oracle (true-label) Landy–Szalay pair counts and hence consistency for the target 2PCF. In controlled simulations with clustered and spatially structured contaminants, PP–LS removes the bias of na\"ive catalog-level estimators while achieving substantially lower variance than spectroscopic-only clustering. The result is a statistically principled, computationally lightweight estimator that integrates directly with standard pair-counting pipelines and supports robust clustering inference in next-generation surveys.
\end{abstract}

\maketitle

\tableofcontents

\section{Introduction}

The two-point correlation function (2PCF) has become a cornerstone of large-scale structure analysis in modern cosmology \citep{peebles2020large,weinberg2013observational}. By quantifying the excess probability of finding galaxies or other sources at a given separation relative to a random distribution, the 2PCF encodes information about primordial fluctuations, structure formation, and the physics of dark matter and dark energy \citep{totsuji1969correlation,de2000flat,estrada2009correlation,huterer2018dark}. Its measurement underlies constraints on the baryon acoustic oscillation (BAO) scale \citep{eisenstein2005detection,percival2010baryon,moon2023first,abbott2024dark,lodha2025desi}, redshift-space distortions \citep{ross20072df,vlah2019exploring}, and galaxy bias \citep{crocce2016galaxy,vakili2023clustering}. With upcoming wide-field surveys such as LSST \citep{chisari2019core}, \emph{Euclid} \citep{amendola2018cosmology}, and the Roman Space Telescope \citep{eifler2021cosmology}, achieving precise and unbiased 2PCF estimates will remain an outstanding challenge for robust cosmological inference.

The go-to method for estimating the 2PCF today is the Landy--Szalay (LS) estimator \citep{landy1993bias}, which improves upon earlier ratio/cross-count forms \citep{hamilton1993toward} and combines data-data, data-random, and random-random pair counts to reduce variance, selection, and edge effects. Random catalogs encode the survey geometry and selection, while per-object/pair weights (e.g., FKP) further improve estimation variance \citep{Feldman1994}. Classical alternatives include the original Peebles' estimator \citep{peebles2020large}, Davis-Peebles \citep{davis1983survey}, Hewett \citep{hewett1982estimation}, Rivolo \citep{rivolo1986two}, and Hamilton \citep{hamilton1993toward}, and other forms \citep{storey2021two}. Systematic comparisons generally find LS (and Hamilton) to be closest to optimal across realistic survey conditions \citep{kerscher2000comparison}. Beyond these binned pair-count estimators, least-squares formulations recast correlation estimation as a linear inference problem: one can interpolate or directly fit $\xi(r)$ at chosen separations with a design matrix built from pairwise kernels, yielding pointwise (or coefficient) estimates with well-defined covariance. Examples include Tessore's least-squares 2PCF \citep{tessore2018least}, which recovers standard estimators as special cases and enables interpolation without hard binning, and the continuous-function estimator of \citep{storey2021two}, which projects pair counts onto smooth basis functions and views LS as the tophat-basis limit \citep{storey2021two}. In practice, these least-squares/generalized estimators can reduce binning artifacts, provide smoother $\xi(r)$ representations, and deliver direct parameter fits (e.g., BAO scale) while retaining the random-catalog normalization that controls edges and selection.

A recurring practical challenge in 2PCF analysis is the reliable construction of a catalog representing the intended target population. In idealized settings, such as high-resolution spectroscopic surveys with unambiguous line identification, membership in the target class can often be determined with high fidelity. In many contemporary surveys, however, target selection relies on imperfect measurements and automated pipelines. Imaging artifacts, flux uncertainties, redshift misestimation, and object-class confusion (e.g., stars or quasars misidentified as galaxies) can all induce misclassification of the intended population \citep{kong2026imaging,davis2023hetdex,johnston2021organised,nicola2020tomographic,crocce2016galaxy,fadely2012star,kim2016star,drlica2018dark}. As a result, catalog membership is frequently represented by \emph{noisy inclusion labels} \citep{crocce2016galaxy}, rather than by error-free indicators of truth-level class membership. Noisy labels introduce two coupled issues: \emph{contamination} (e.g., non-galaxies included) and \emph{incompleteness} (e.g., true galaxies excluded). Both modify the effective selection function and can bias clustering measurements \citep{awan2020angular,farrow2021correcting}, smearing cosmological signals \citep{chaves2018effect}.

In wide-area surveys, contamination and incompleteness arise from multiple sources, including star-galaxy classification error, photometric-redshift failures, and related selection effects. These imperfections are seldom spatially uniform. Stellar density changes with Galactic latitude, extinction traces the structure of Galactic dust, and survey depth depends on observing conditions such as seeing, airmass, and sky brightness \citep{kong2026imaging}. As a consequence, misclassification rates and the associated contamination and incompleteness vary coherently across the sky. This spatial structure propagates directly into clustering measurements: elevated stellar contamination in dense stellar fields can introduce spurious large-scale power, whereas incompleteness correlated with shallow depth or degraded image quality can attenuate the true clustering amplitude. Simple global purity corrections cannot remove these position-dependent effects; robust inference instead requires methods that explicitly model or correct for spatially varying, noisy labels \citep{rodriguez2022dark}.

A range of strategies has been developed to mitigate these sources of estimation bias. Template regression and mode deprojection can downweight known systematic contamination biases \citep{leistedt2014exploiting}. At the object level, reweighting and veto masks are used in BOSS and DES analyses to excise or correct problematic regions \citep{ross2012clustering,elvin2018dark}. Masking and aggressive quality cuts, however, reduce statistical power and may introduce selection effects of their own. Regression against survey-property templates depends on the availability and fidelity of those templates, which is difficult to model and verify, and cannot remove clustering intrinsic to contaminants \citep{rodriguez2022dark}. Weighting objects by $p(\mathrm{galaxy}\mid\text{features})$ can lower variance, but unbiasedness requires perfectly calibrated probabilities and no spatial structure in the residual errors, assumptions that are seldom met. Spectroscopic calibration via cross-comparisons with photometric samples can help \citep{newman2008calibrating}, yet such approaches typically rely on strong assumptions about selection probabilities, incompleteness, and impurity rates.

An alternative is forward modeling: explicitly include contamination and incompleteness in the theoretical prediction and fit to the data \citep{crocce2016galaxy}. While principled, forward modeling can be computationally expensive and requires accurate models of multiple systematics, which is often infeasible in practice. A shared limitation of these strategies is their reliance on strong, often unverifiable or difficult-to-verify assumptions (e.g., calibrated probabilities, known completeness/purity, stationary error models, or correctly specified systematic templates). The effort required to validate and maintain these assumptions, through extensive calibration, cross-matching, and null tests, can exceed the primary analysis, and even then, residual sensitivity to mismodeling often remains.

These limitations motivate an assumption-lean framework that (i) preserves the geometric and edge-correction advantages of LS-style estimators, (ii) accepts arbitrary noisy classifiers without assuming calibration, known contamination details, or even origins of contamination, and (iii) fully exploits a small, high-fidelity spectroscopic sample available in most surveys. Our guiding principle is to incorporate the correction directly into the estimator itself, rather than relying on explicit forward modeling or post-hoc adjustments and calibration. The procedure does not require prior identification of specific systematics, explicit probability calibration, global purity factors, hand-crafted systematics templates, or computationally intensive forward simulations. Instead, the estimator is constructed so that bias arising from contamination and incompleteness is removed internally through its design, yielding a self-contained and design-consistent correction mechanism. The only requirement is a small spectroscopic gold-standard sample. 

We propose a \emph{prediction-powered} Landy-Szalay (PP-LS) estimator, adapting prediction-powered inference (PPI) \citep{angelopoulos2023ppi} to the pairwise U-statistics that define the 2PCF. PPI offers a general recipe for combining noisy source classification with a small labeled set to obtain valid population-level inference. Our construction starts with a plug-in LS estimator that uses noisy labels (hard or soft) and then adds residual-based correction terms computed only on a small spectroscopic subset. By upweighting these residuals with standard survey-sampling ideas, we recover unbiased estimates of the true pair counts, regardless of classifier calibration, accuracy, or spatial structure, while remaining fully compatible with random catalogs and the standard LS geometry corrections. Notably, the approach requires only a handful of additional weighted pair counts beyond LS, making it computationally lightweight and easy to integrate into existing codes such as \texttt{TreeCorr} and \texttt{Corrfunc}.

\begin{figure*}[th!]
    \centering
    \includegraphics[width=1\linewidth]{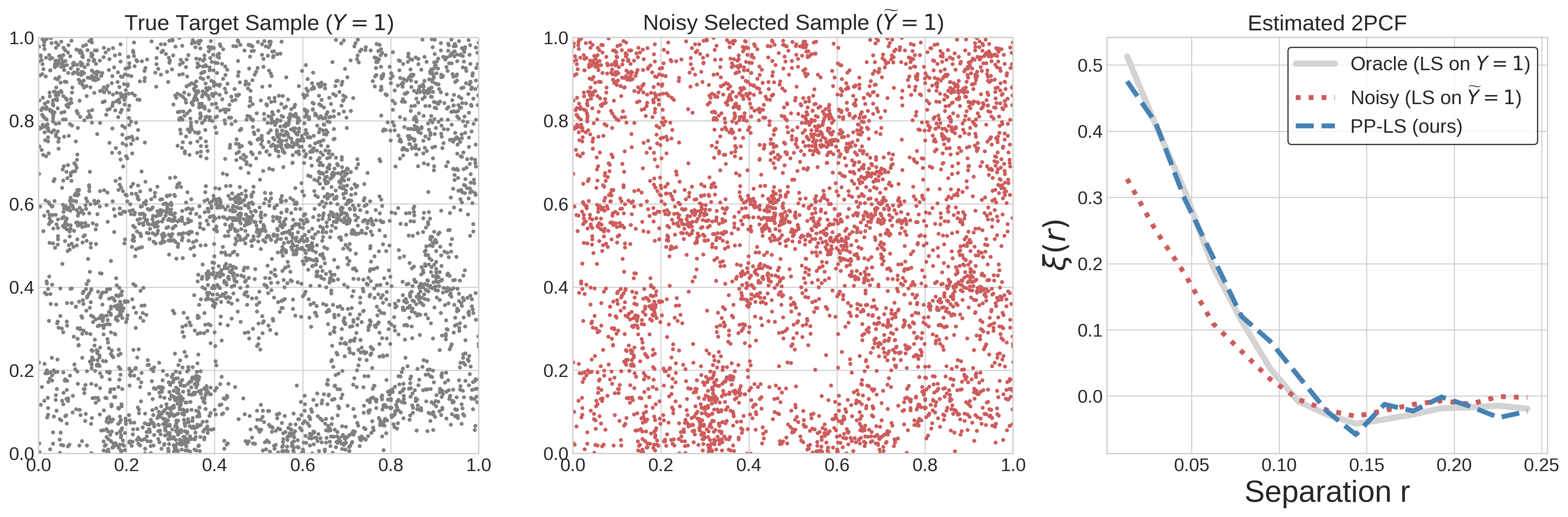}
    \caption{{\bf Left.} Simulated true target sources drawn from a clustered Thomas process in a unit square. {\bf Middle.} Objects with $\widetilde{Y}=1$, mimicking noisy source catalog with 30\% contamination (false positive). Spatially varying contamination produces a hotspot and gradient that biases pair counts if uncorrected. {\bf Right.} Estimated 2PCF $\xi(r)$. Oracle (LS on $Y$), Noisy (LS on $\widetilde{Y}$), and PP--LS method. Noisy LS is biased; PP--LS aligns with the oracle by using residual corrections from a small spectroscopic subset.}
  \label{fig:toy_example}
\end{figure*}

\paragraph*{Illustrative Example (Figure~\ref{fig:toy_example}).} To illustrate the problem and our solution, we simulate a two-dimensional field in a unit square. Galaxy positions are drawn from a clustered Thomas process. Contaminants follow an inhomogeneous process with a large-scale gradient plus a hotspot (see Figure~\ref{fig:contamination_error_fields}). We generate a noisy label $\widetilde{Y}\in \{0, 1\}$ with spatially varying false-positive/negative rates (higher near the hotspot and within a shallow-depth stripe), yielding a contamination of $\sim$30\% overall. A small spectroscopic subset (10\% of data) provides gold labels $Y$. The LS estimator built on $\widetilde{Y}$ is biased high at small separations due to clustered contaminants and biased on large scales by the gradient. Our prediction-powered LS (PP-LS) corrects these biases by adding residual-based pair-count adjustments estimated only on the spectroscopic subset, closely tracking the oracle (true-label) LS across scales. 

The remainder of this paper presents the method and benchmarks it. Sec.~\ref{sec:prereq} reviews prerequisites and introduces the two-point correlation function. Sec.~\ref{sec:estimator} formalizes the problem and derives our proposed estimator. In Sec.~\ref{sec:benchamrk}, we describe the benchmarking and evaluation procedure; then, we illustrate the proposed estimator's performance and benchmark it against alternative decontamination methods in Sec.~\ref{subsec:results}. Finally, we conclude in Sec.~\ref{sec:conclusion}.

\section{Prerequisites and Definitions}
\label{sec:prereq}

Our target throughout is 2PCF of astrophysical objects with a well-defined selection. In clustering analyses, the quantity of interest is rarely the full detection catalog. Instead, one aims to measure 2PCF of a \emph{physically defined population}; for example, galaxies in a given redshift interval, objects above a luminosity threshold, emission-line sources of a particular type, or any other class specified by intrinsic properties. In practice, however, catalog membership is determined from noisy measurements and imperfect pipelines. The central distinction in this work is therefore between \emph{true inclusion in the target population} and \emph{observed inclusion in the working catalog used for pair counting}.

\subsection{Target Population and Inclusion Indicator} \label{sec:objects_def}

\paragraph*{True inclusion label $Y$.}
Let $\mathcal{T}$ denote a fixed target selection rule defined on the \emph{true} (latent) physical properties of an object. For each detected object $i\in\{1,\dots,n\}$ at position $s_i$, define the \emph{true inclusion indicator}
\begin{equation}
Y_i \in \{0,1\},
\end{equation}
where 
\begin{equation*}
Y_i =
\begin{cases}
1 & \text{if object } i \text{ satisfies the target rule } \mathcal{T}, \\
0 & \text{otherwise.}
\end{cases} 
\end{equation*}

The point process whose clustering we seek to estimate is
\[
G := \{ s_i : Y_i = 1 \}.
\]
All occurrences of $\xi$ in this paper refer to the 2PCF of this truth-level process $G$, that is, correlations among pairs with $Y_i=Y_j=1$. We note that $Y_i$ reflects an ideal inclusion decision based on true physical quantities, irrespective of how the observed catalog was constructed.

\paragraph*{Noisy inclusion label $\widetilde Y$.}
In practice, in (imaging) surveys,  $Y_i$ is unknown for the vast majority of objects. Instead, the analyst works with a catalog generated from measured fluxes, estimated redshifts, classification pipelines, and quality cuts. We represent the resulting working-sample membership by a \emph{noisy inclusion indicator}
\begin{equation}
\widetilde Y_i \in [0,1].
\end{equation}
This variable indicates whether, and to what extent, object $i$ is included in the sample used for pair counting:
\begin{itemize}
\item In the \emph{hard-selection} case, $\widetilde Y_i \in \{0,1\}$ indicates binary inclusion or exclusion by the pipeline.
\item In the \emph{soft-selection} case, $\widetilde Y_i \in [0,1]$ may represent a model score or estimated probability of inclusion, which reflects the uncertainty of the label. 
\end{itemize}
Even if $\widetilde Y_i$ is probabilistic and approximately calibrated \citep{vashistha2025mathcali}, replacing $Y_i$ by $\widetilde Y_i$ in the Landy--Szalay estimator generally yields a biased estimate of $\xi$, particularly when misclassification errors are spatially structured. Throughout the remainder of this work, we refer to $Y$ and $\widetilde {Y} $ as true and noisy labels, respectively. 

\paragraph*{Examples.} The abstract definition of $Y$ encompasses common cosmological use cases:
\begin{itemize}
\item \textbf{Tomographic bin selection.} Suppose the scientific goal is to measure the 2PCF of galaxies with true redshift $z\in[z_{\min},z_{\max}]$ and true luminosity $L\ge L_{\min}$ \citep{weaverdyck2026dark}. Then $Y_i=1$ if object $i$ satisfies these conditions based on its true redshift and intrinsic luminosity.
\item \textbf{Emission-line survey.} In a survey targeting Lyman-$\alpha$ emitters \citep{davis2023hetdex}, $Y_i=1$ if object $i$ is truly an LAE within the desired redshift range, as determined by its intrinsic line properties.
\item \textbf{General population definition.} More broadly, $Y_i$ may encode membership in any physically defined class whose clustering is of interest.
\end{itemize}

Discrepancies between $Y_i$ and $\widetilde Y_i$ arise from multiple sources:
\begin{itemize}
\item \textbf{Interloper contamination.} Quasars or stars misclassified as galaxies may satisfy pipeline cuts, yielding $\widetilde Y_i=1$ even though $Y_i=0$.
\item \textbf{Scatter across selection thresholds.} Measurement noise in luminosity or flux can move objects across selection boundaries, producing both false positives ($Y_i=0$, $\widetilde Y_i=1$) and false negatives ($Y_i=1$, $\widetilde Y_i=0$).
\item \textbf{Redshift assignment errors.} Photometric redshift errors or low-resolution spectroscopic failures can place objects in the wrong tomographic bin, again leading to $Y_i\neq\widetilde Y_i$ \citep{ma2006effects}.
\item \textbf{Spatially correlated systematics.} Imaging artifacts, depth variations, or environment-dependent redshift failures can induce spatial patterns in $\widetilde Y$ errors, complicating standard debiasing approaches.
\end{itemize}

In galaxy imaging surveys, photometric redshift errors provide the dominant mechanism by which $\widetilde Y_i$ may differ from $Y_i$, particularly in tomographic analyses where bin membership is defined by redshift intervals \citep{ma2006effects,zhang2025forecasting}. However, within the formalism developed here, photo-$z$ scatter is only one instance of a broader class of misclassification error: a noisy inclusion indicator $\widetilde Y_i$ that may deviate from the truth-level label $Y_i$ through measurement error, classification uncertainty, or pipeline decisions, potentially in a spatially correlated manner. We formulate the estimator entirely in terms of the pair $(Y_i,\widetilde Y_i)$, without specifying a generative or forward model for photometric redshifts or any other contamination mechanism. Any process that moves objects across the boundary of the target set $G$ (for example, redshift misestimation that shifts galaxies across tomographic bin edges, luminosity scatter across thresholds, or low-resolution spectroscopic failures) manifests mathematically as a mismatch between $Y_i$ and $\widetilde Y_i$. Our proposed estimator operates directly on this mismatch and, therefore, applies uniformly across these scenarios without requiring explicit modeling of the underlying error source.

Finally, throughout this work, we implicitly take the target population to be galaxies. However, the formalism is entirely general and applies to any class of astrophysical objects defined by a binary inclusion rule (e.g., quasars, galaxy clusters, stars, star clusters, or compact-object mergers). For concreteness and notational simplicity, we present the methodology in the context of galaxies.

\subsection{Two-point Correlation Function}

Recall $G=\{s_i: Y_i=1\}$ denote the (unknown) set of sources satisfies the target rule $\mathcal{T}$ drawn from a point process on a domain $\mathcal{S}\subseteq\mathbb{R}^3$ (for 3D analyses) or $\mathcal{S}=\mathbb{S}^2$ (for angular analyses). Define the true galaxy counting measure
\begin{equation}
N_G(A) = \sum_{i=1}^n Y_i\,\mathds{1}\{s_i\in A\}, \qquad A\subseteq \mathcal{S},
\end{equation}
with (first) intensity $\lambda_G(s)$ satisfying $\mathbb{E}[N_G(ds)]=\lambda_G(s)\,ds$.
The \emph{two-point correlation function} $\xi$ quantifies the excess probability of finding a galaxy pair at a given separation. 

Let $d(\cdot,\cdot)$ denote the separation metric. For 3D comoving, $\xi(r)$, if $s_i$ includes comoving coordinates (e.g., from redshifts, distances), we have
\begin{equation}
d_{ij} = \|s_i - s_j\| \in \mathbb{R}_+.
\end{equation}
Let $w(\theta)$ be angular 2PCF on the sphere. For unit vectors $\hat s_i$ and $\hat s_j$, we have
\begin{equation}
d_{ij} = d(s_i,s_j) = \arccos(\hat s_i\cdot\hat s_j) \in [0,\pi].
\end{equation}

\paragraph*{3D definition (pair probability).}
For distinct volume elements $ds_1,ds_2$ separated by $r=\|s_1-s_2\|$,
\begin{equation}
\label{eq:xi-prob}
\mathbb{E}\!\left[dN_G(s_1)\,dN_G(s_2)\right]
=
\lambda_G(s_1)\lambda_G(s_2)\,\bigl[1+\xi(r)\bigr]\,ds_1\,ds_2.
\end{equation}
Under (approximate) homogeneity and isotropy on the scales of interest, $\xi$ depends only on separation $r$. In Equation~\eqref{eq:xi-prob}, $g(r):=1+\xi(r)$ is the pair-correlation function of the point process \citep{satoh2003introduction}.

\paragraph*{3D definition (overdensity field).}
Let $n_G(s)$ be the galaxy number density and $\bar n$ a reference mean (e.g., the selection-averaged mean inside the survey window). Define the overdensity $\delta_G(s) := \tfrac{n_G(s)-\bar n}{\bar n}$. Then
\begin{equation}
\label{eq:xi-cov}
\xi(r) = \big\langle \delta_G(s)\,\delta_G(s+r)\big\rangle,
\end{equation}
where $\langle\cdot\rangle$ denotes an ensemble (or an ergodic spatial) average.

\paragraph*{Angular correlation and projections.}
For surface density $\Sigma(\hat{\boldsymbol{s}})$ on the sphere and surface overdensity  $\delta_\Omega(\hat{\boldsymbol{s}}):=\tfrac{\Sigma(\hat{\boldsymbol{s}})-\bar\Sigma}{\bar\Sigma}$, we have
\begin{equation}
w(\theta) = \big\langle \delta_\Omega(\hat{\boldsymbol{s}})\,\delta_\Omega(\hat{\boldsymbol{s}}')\big\rangle,
\end{equation}
where $\theta=\arccos(\hat{\boldsymbol{s}}\!\cdot\!\hat{\boldsymbol{s}}')$. Projected statistics such as $w_p(r_p)$ are related to $\xi$ by line-of-sight integration
\begin{equation}
w_p(r_p) = 2\!\int_{0}^{\pi_{\max}}\!\!\xi\!\Big(\sqrt{r_p^2+r_{\perp}^2}\,\Big)\,dr_{\perp},
\end{equation}
with $r_p$ the transverse separation and $r_{\perp}$ the line-of-sight separation. For notation convenience, we focus on $\xi$, but the results generalize to other 2PCFs as well. 

\paragraph*{Mask, selection, and the bin-averaged target.}
Empirically, analyses are performed within a finite survey window $W:\mathcal{S}\to\{0,1\}$ (or $W\in[0,1]$ for fractional completeness), with a possibly varying selection $\lambda_G(s)$.
Given a separation metric $d(\cdot,\cdot)$ and a bin kernel $K_b$ (see Sec.~\ref{sec:kernel_def} for definition), the \emph{bin-averaged} correlation we target is
\begin{widetext}
\begin{equation}
\label{eq:xi-bin-avg}
\xi_b
:=
\frac{\displaystyle \iint_{\mathcal{S}\times\mathcal{S}}
W(s_1)W(s_2)\,\lambda_G(s_1)\lambda_G(s_2)\,
\xi\!\big(d(s_1,s_2)\big)\,K_b\!\big(d(s_1,s_2)\big)\,ds_1\,ds_2}
{\displaystyle \iint_{\mathcal{S}\times\mathcal{S}}
W(s_1)W(s_2)\,\lambda_G(s_1)\lambda_G(s_2)\,
K_b\!\big(d(s_1,s_2)\big)\,ds_1\,ds_2}.
\end{equation}
\end{widetext}
In practice, the denominator is estimated by random catalogs matching $W$ and $\lambda_G$, and the same kernel $K_b$ is used for data-data, data-random, and random-random to preserve the Landy--Szalay identity (see Sec.~\ref{app:LS_derivation}).

We recall $\xi_b$ is a kernel- (or bin-)averaged version of the point-wise correlation $\xi(r)$. $\xi_b \approx \xi(r)$ whenever $\xi$ varies slowly across the bin support of $K_b$ (e.g., for narrow bins evaluated $r$). Moreover, in the idealized narrow-bin limit where the kernel collapses to a Dirac delta shell at separation $r$,
$$
K_b \big(d(s_1,s_2)\big) \to \delta_D \big(r-d(s_1,s_2)\big),
$$
and under the usual assumption that $\xi$ depends only on separation (statistical homogeneity and isotropy), the ratio in Equation~\eqref{eq:xi-bin-avg} reduces exactly to $\xi_b=\xi(r)$: the common window/selection factors in numerator and denominator cancel, while $\xi\big(d(s_1,s_2)\big)$ is constant on the $\delta_D$-selected shell $d(s_1,s_2)=r$. For the remainder of this work, we use the notation $\xi$ exclusively to denote 2PCF, and reserve $w$ for the weight function, which will be defined in a subsequent section.

\section{Two-Point Correlation Function Estimator}
\label{sec:estimator}

This section formalizes the estimation problem for 2PCF in the presence of noisy inclusion labels. We precisely define the truth-level population whose clustering is the scientific target, introduce the oracle LS estimator that would be available if all true labels were observed, and then characterize the practical setting in which only noisy labels are available for the full sample while exact labels are known for a small random subset. We develop a prediction-powered correction that reconstructs the oracle pair counts using residual information from the labeled subset, establish its design-based properties under simple random labeling, and describe its implementation with standard weighted pair-counting routines. 

\subsection{Problem Setup and Assumptions}
\label{sec:setup}

We briefly recall the notation for inclusion labels and specify the statistical object to be estimated. The scientific target is the 2PCF of the truth-level population
$G = \{ s_i : Y_i = 1 \}$, where $Y_i\in\{0,1\}$ is the true inclusion indicator defined by a fixed physical selection rule $\mathcal{T}$ (see Sec.~\ref{sec:objects_def}). Throughout, $\xi$ denotes the 2PCF of this process, i.e., correlations among pairs with $Y_i=Y_j=1$.

\subsubsection*{Source Catalog}

Suppose we observe $n$ detected objects indexed by $i=1,\dots,n$, each located at position $s_i$ (either on the sky or in three-dimensional comoving space). Let
\begin{equation}
S := \{s_i : i=1,\dots,n\}    
\end{equation}
denote the full set of detected objects within the survey footprint. 

Recall $\mathcal{T}$ is a fixed selection rule defined on the true (latent) physical properties of objects and $Y_i\in\{0,1\}$ denote the corresponding true inclusion indicator. The truth-level target population is then defined mathematically as
\begin{equation}
G :=  \{ s_i \in S : Y_i = 1 \}.    
\end{equation}
In practice, however, $Y_i$ is unknown for most objects. The working catalog used for clustering analysis is typically constructed from noisy measurements and selection criteria applied at the pipeline level (see Sec.~\ref{sec:objects_def}). Let
\begin{equation}
\widetilde G := \{ s_i \in S : \widetilde Y_i = 1 \}    
\end{equation}
denote the observed inclusion set induced by the noisy labels $\widetilde Y_i$. The standard LS estimator applied directly to $\widetilde Y$ targets the clustering of the observed catalog $\widetilde G$ rather than the truth-level process $G$, and is generally biased for $\xi$ whenever $\widetilde Y_i \neq Y_i$ for a non-negligible fraction of objects.

The discrepancy between $G$ and $\widetilde G$ can be decomposed into two components: a contamination set 
\begin{equation}
C := \{ s_i \in A : \widetilde Y_i = 1,\, Y_i = 0 \},    
\end{equation}
consisting of objects that are included in the analysis but do not belong to the target population, and a missed set 
\begin{equation}
M := \{ s_i \in A : \widetilde Y_i = 0,\, Y_i = 1 \},    
\end{equation}
consisting of true target objects that are excluded. The estimator developed below corrects the pair counts so that the resulting 2PCF targets $G$, even though the observed catalog $\widetilde G$ may differ from $G$ through contamination and incompleteness.

\subsubsection*{Random Catalog}

Recall $W:\mathcal{S}\to\{0,1\}$ denote the survey window (or, more generally, $W\in[0,1]$ for fractional completeness), and let $\lambda_{\mathrm{obs}}(s)$ denote the effective selection function of the detected source catalog $S=\{s_i\}_{i=1}^n$, incorporating angular mask, depth variations, radial selection, and other large-scale completeness effects. The observed detections may be viewed as a realization of an inhomogeneous point process on $\mathcal{S}$ with intensity proportional to $\lambda_{\mathrm{obs}}(s)W(s)$, modulated by clustering.

A random catalog is a synthetic point set
\begin{equation}
R=\{s_a^R\}_{a=1}^{n_R},    
\end{equation}
whose locations are drawn independently from a distribution with density proportional to $\lambda_{\mathrm{obs}}(s)W(s)$, but without clustering (i.e., as an inhomogeneous Poisson process with that intensity). In particular, for measurable $A\subseteq\mathcal{S}$,
\[
\mathbb{E}[N_R(A)]
=
n_R \,
\frac{\int_A \lambda_{\mathrm{obs}}(s)W(s)\,ds}
{\int_{\mathcal{S}} \lambda_{\mathrm{obs}}(s)W(s)\,ds},
\]
where $N_R(A)$ denotes the number of random points falling in $A$.

The role of $R$ in the LS-style estimator construction is to provide an empirical representation of the separable baseline intensity $\lambda_{\mathrm{obs}}(s)W(s)$, thereby removing leading-order effects of survey geometry and large-scale selection gradients. 

The prediction-powered correction developed in Sec.~\ref{sec:pp_LS} operates strictly at the level of inclusion labels, replacing $Y_i$ by design-unbiased estimators based on $(\widetilde Y_i,Y_i)$. It does not modify the role of the random catalog. In particular, even if inclusion misidentification were perfectly corrected, an incorrect specification of $\lambda_{\mathrm{obs}}(s)W(s)$ in the random catalog would induce bias in $\widehat{\xi}(b)$ through mismodeling of survey geometry or baseline selection. Thus, our proposed estimator should be understood as augmenting, rather than replacing, the standard LS requirement that the random catalog faithfully represent the observational selection of the source catalog $S$.

\subsubsection*{Labeled subset}

We assume access to a labeled index set 
\begin{equation}
L \subset \{1,\dots,n\},  \qquad |L| = m \ll n,    
\end{equation}
for which the true inclusion labels $\{Y_i : i \in L\}$ are observed without error (e.g., via high-resolution, high-fidelity spectroscopic confirmation). 

We model $L$ as a \emph{simple random sample without replacement} from the source catalog indices $\{1,\dots,n\}$. That is, conditional on the realized catalog $\{(s_i,\widetilde Y_i,Y_i)\}_{i=1}^n$, the labeled set $L$ is uniformly distributed over all subsets of size $m$. Equivalently, for any $i\neq j$,
\[
\Pr(i \in L) = \frac{m}{n},
\qquad
\Pr(i,j \in L) = \frac{m(m-1)}{n(n-1)},
\]
and these inclusion probabilities do not depend on $s_i$, $Y_i$, $\widetilde Y_i$, or any other object-level characteristics. Expectations in the theoretical analysis are taken with respect to this sampling design. This design-based assumption implies that the labeled subset is statistically representative of the full source catalog $S$ in the sense that, for any fixed function $g(i)$,
\begin{equation}
\mathbb{E}_L\!\left[\frac{n}{m}\sum_{i\in L} g(i)\right] = \sum_{i=1}^n g(i),    
\end{equation}
and similarly for pairwise functions with the appropriate finite-population correction. 

The simple-random-sampling assumption is fundamental for the unbiasedness results established below and a high-fidelity spectroscopist sample must satisfy this condition. It guarantees that inclusion in $L$ is independent of spatial position and object properties, so that the labeled objects provide an unbiased view of the discrepancy between $Y$ and $\widetilde Y$ across the survey footprint. If, instead, the labeled subset were preferentially drawn from specific regions, magnitudes, environments, or redshift ranges, then the sampling mechanism would be correlated with $s_i$ or with $(Y_i,\widetilde Y_i)$. In that case, the Horvitz--Thompson scaling used in the estimator would require modification to incorporate unequal inclusion probabilities or an explicit model for the labeling design.

On labeled objects, define the observed residual
\begin{equation}
\Delta_i := Y_i - \widetilde Y_i, \qquad i \in L.
\end{equation}
The residual $\Delta_i$ measures the deviation between truth-level and catalog-level inclusion for object $i$. The key idea of the prediction-powered construction is to use empirical averages of these residuals, together with design-based scaling, to reconstruct the pair counts that would have been obtained had all $Y_i$ been observed.

\subsection{Binning Kernels $K_b$} \label{sec:kernel_def}

We consider a finite set of separation bins indexed by $b\in\mathcal{B}$. For each bin, we define a \emph{kernel}
\begin{equation}
K_b:\mathbb{R}_+\to\mathbb{R}_+,\qquad r\mapsto K_b(r),
\end{equation}
that assigns weight to a pair with separation $r=d_{ij}$.  Here, $K_b(\cdot)$ denotes the kernel (bin) indicator for separation bin $b$. To recover standard histogram binning, choose a \emph{top-hat} kernel
\begin{equation}
\label{eq:tophat}
K_b^{\text{TH}}(r) := \mathds{1}\{r\in [r_{b,\min},\, r_{b,\max})\}.
\end{equation}
Other smooth choices (sometimes advantageous for variance reduction or differentiability) include:
\begin{align}
\text{Triangular:}\quad
&K_b^{\triangle}(r) := \max\!\left(1-\frac{|r-c_b|}{h_b},\,0\right),\\
\text{Epanechnikov:}\quad
&K_b^{\text{Epa}}(r) := \max\!\left(1-\Big(\tfrac{r-c_b}{h_b}\Big)^2,\,0\right),\\
\text{Gaussian:}\quad
&K_b^{\mathcal{N}}(r) := \exp\!\left(-\frac{(r-c_b)^2}{2h_b^2}\right).
\end{align}
Here $c_b$ denotes a bin center and $h_b$ a half-width (or bandwidth); for $w(\theta)$ one sets $r\mapsto\theta$. Choice of bandwidth $h_b$ controls the bias-variance trade-off. The same kernel must be used for data-data ($DD$), data-random ($DR$), and random-random ($RR$) counts to preserve the Landy--Szalay identity.

For $w(\theta)$, a common choice of agular bin is logarithmically spaced edges $\{\theta_{b,\min},\theta_{b,\max}\}$ with $K_b^{\text{TH}}(\theta)$ from Eq.~\eqref{eq:tophat}. For $w_p(r_p)$, one may adopt a top-hat in $r_p$ and a top-hat in line-of-sight $r_{\perp}$ with $|r_{\perp}|\le r_{\perp, \max}$, i.e.\ $K_b(r_p, r_{\perp})=\mathds{1}\{r_p\in b\}\mathds{1}\{|r_{\perp}|\le r_{\perp, \max}\}$, and then sum over $r_{\perp}$. 

In our experiments, we adopt a simple top-hat filter, primarily for interpretability and consistency with standard practice. The proposed estimator, however, is not restricted to this choice and extends directly to alternative kernel functions \citep{storey2021two}. Exploring how the filter shape and bandwidth can be optimized as a function of survey characteristics (e.g., number density, footprint geometry, and expected clustering scale) could be a promising direction for future work. In particular, adaptive or smooth kernels may offer improved bias-variance trade-offs in regimes where shot noise, masking, or label sparsity become limiting factors.

\subsection{Pair Weights and Random Catalogs}

Let $w_{ij}\ge 0$ denote symmetric pair weights for data--data pairs. These may encode, for example, FKP weights \citep{Feldman1994}, imaging-systematics weights, or products of per-object weights \citep{pearson2016optimal}. In an FKP-like scheme, if $n(\mathbf{x})$ denotes the expected number density and $P_0$ is a fiducial power-spectrum amplitude, one may define per-object weights $u_i \propto \bigl[1 + n(\mathbf{x}_i) P_0\bigr]^{-1}$ and then $w_{ij} = u_i u_j$.  Analogous weights $w^R_{ia}$ and $w^R_{ab}$ are used for data--random and random--random pairs. We retain explicit weights throughout the theoretical development so that the results apply to arbitrary non-negative weighting schemes.

In the numerical experiments presented later, we set all weights identically equal to unity in order to isolate the impact of label misidentification. The theoretical formulation, however, accommodates general weighting schemes without modification.

\subsection{Target Counts and Landy--Szalay}

Allow general symmetric pair weights, with $w_{ij} = w_{ji}$ for $DD$ pairs, $w^{R}_{ia}$ for $DR$ pairs, and $w^{R}_{ab} = w^{R}_{ba}$ for $RR$ pairs. 

Define the \emph{label-aware denominators}
\begin{align}
S_2 &:= \sum_{1\le i<j\le n} w_{ij}\,Y_i Y_j, \\
S_{1R} &:= \sum_{i=1}^n \sum_{a=1}^{n_R} w^R_{ia}\,Y_i, \\
R_2 &:= \sum_{1\le a<b\le n_R} w^R_{ab}.
\end{align}
The \emph{normalized} histograms are
\begin{align}
DD(b) &:= \frac{\displaystyle \sum_{1\le i<j\le n} w_{ij}\,Y_i Y_j\,K_b(d_{ij})}{S_2},\\
DR(b) &:= \frac{\displaystyle \sum_{i=1}^n \sum_{a=1}^{n_R} w^R_{ia}\,Y_i\,K_b(d_{ia})}{S_{1R}},\\
RR(b) &:= \frac{\displaystyle \sum_{1\le a<b\le n_R} w^R_{ab}\,K_b(d_{ab})}{R_2}.
\end{align}
The Landy--Szalay estimator \citep{landy1993bias} is
\begin{equation}
\widehat \xi(b) = \frac{DD(b) - 2\,DR(b) + RR(b)}{RR(b)}.
\end{equation}

We assume that the random points are drawn independently from a distribution proportional to the survey window and selection function of the observed catalog, and that the same binning kernel $K_b$ and weighting scheme are applied consistently to $DD$, $DR$, and $RR$ pair counts. 

Under standard regularity conditions, namely that (i) $R$ correctly samples $\lambda_{\mathrm{obs}}(s)W(s)$, (ii) identical binning kernels $K_b$ and weighting schemes are applied to $DD$, $DR$, and $RR$ terms, and (iii) $n_R$ is sufficiently large that random-pair shot noise is negligible, the oracle LS estimator computed with true labels $Y$ satisfies
\begin{equation}
\mathbb{E}\big[\widehat{\xi}^\star(b)\big] \approx \xi(b),  
\end{equation}
where the expectation is taken over realizations of the underlying clustered point process (see Appendix~\ref{app:LS_derivation} for derivations and details). Our proposed estimator, described in Sec.~\ref{sec:pp_LS},  inherits this property after correcting for label noise.

\subsection{Prediction-powered Decomposition and 2PCF Estimator} \label{sec:pp_LS}

Next, we use the prediction-powered inference \cite{angelopoulos2023ppi} recipe to develop an estimator of the 2PCF from noisy labels. We recall observed residuals on the labeled subset $L$ is
\begin{equation}
\Delta_i := Y_i-\widetilde Y_i, \quad i\in L;
\end{equation}
and for a bookkeeping convention used in implementation  (see Sec.~\ref{sec:implimentation}), we set observed residuals to zero on the unlabeled subset
\begin{equation}
     \Delta_i := 0,\quad i\notin L.
\end{equation}

The identity
\begin{equation}
Y_i Y_j \equiv \widetilde Y_i\widetilde Y_j + \Delta_i \widetilde Y_j + \widetilde Y_i \Delta_j + \Delta_i \Delta_j
\end{equation}
holds for both hard and soft $\widetilde Y$. For algebraic convenience, we define \emph{ordered-pair} sums over $i\neq j$
\begin{align}
T_0(b)  &:= \sum_{i\neq j} w_{ij}\,\widetilde Y_i\widetilde Y_j\,K_b(d_{ij}),\\
T_1(b)  &:= \sum_{i\neq j} w_{ij}\,\Delta_i \widetilde Y_j\,K_b(d_{ij}),\\
T_2(b)  &:= \sum_{i\neq j} w_{ij}\,\Delta_i\Delta_j\,K_b(d_{ij}).
\end{align}
Converting to \emph{unordered} pairs introduces a factor $1/2$; thus the (target) galaxy--galaxy numerator is
\begin{equation}
{DD}(b)\;=\;\tfrac12\!\left[T_0(b)+2T_1(b)+T_2(b)\right].
\end{equation}

Since only a subset of data are labeled, we estimate the terms involving $\Delta$ by Horvitz--Thompson scaling for simple random labeling (sample size $m$ from $n$)
\begin{align}
\widehat T_0(b)   &:= T_0(b),\\
\widehat T_1(b)   &:= \frac{n}{m}\sum_{i\in L}\sum_{j\ne i} w_{ij}\,\Delta_i\widetilde Y_j\,K_b(d_{ij}),\\
\widehat T_2(b)   &:= \frac{n(n-1)}{m(m-1)}\sum_{\substack{i,j\in L\\ i\ne j}} w_{ij}\,\Delta_i\Delta_j\,K_b(d_{ij}).
\end{align}
The \emph{unordered} PP numerator in bin $b$ is
\begin{equation}
\widehat{DD_\mathrm{PP}}(b)
\;:=\; \tfrac12\!\left[\widehat T_0(b)+\widehat 2T_1(b)+\widehat T_2(b)\right].
\end{equation}

To match the subset-style normalization, we also estimate the denominators
\begin{align}
\widehat S_2 &:= \tfrac12\!\left[\widehat T_0^{(\mathrm{tot})}+\widehat 2T_1^{(\mathrm{tot})}+\widehat T_2^{(\mathrm{tot})}\right],
\end{align}
where the ``totals'' are obtained by replacing $K_b(d_{ij})$ with $1$ in the definitions of $\widehat T_k(b)$ for $k\in\{0,1,2\}$.
Hence the \emph{normalized} prediction-powered $DD$ term is
\begin{equation}
\widetilde{DD}_{\mathrm{PP}}(b) \;:=\; \frac{\widehat{DD_\mathrm{PP}}(b)}{\widehat S_2}.
\end{equation}

The $DR$ term is linear in $Y$ and needs only first-order correction:
\begin{widetext}
\begin{align}
\widehat{DR_\mathrm{PP}}(b)
&:= \sum_{i=1}^n \sum_{a=1}^{n_R} w^R_{ia}\,\widetilde Y_i\,K_b(d_{ia})
   \;+\; \frac{n}{m}\sum_{i\in L}\sum_{a=1}^{n_R} w^R_{ia}\,\Delta_i\,K_b(d_{ia}),\\
\widehat S_{1R} &:= \sum_{i=1}^n \sum_{a=1}^{n_R} w^R_{ia}\,\bigl(\widetilde Y_i + \tfrac{n}{m}\,\Delta_i\bigr),
\qquad
\widetilde{DR}_{\mathrm{PP}}(b) \;:=\; \frac{\widehat{DR_\mathrm{PP}}(b)}{\widehat S_{1R}}.
\end{align}
\end{widetext}
$RR$ is unchanged and normalized as in the original LS estimator
\begin{equation}
\widetilde{RR}(b) \;=\; \frac{\sum_{1\le a<b\le n_R} w^R_{ab}\,K_b(d_{ab})}{R_2}.
\end{equation}
Finally, the prediction-powered LS (PP-LS) estimator is
\begin{equation}
\boxed{
\widehat \xi_{\mathrm{PP}}(b)
= \frac{\widetilde{DD}_{\mathrm{PP}}(b) - 2\,\widetilde{DR}_{\mathrm{PP}}(b) + \widetilde{RR}(b)}{\widetilde{RR}(b)}
} 
\end{equation}
which has the same form as the original LS estimator. 

\subsection{Theoretical Properties}
\label{subsec:theory}

We formalize the relationship between the prediction-powered estimator and the oracle LS estimator that would be computed if the true labels $\{Y_i\}_{i=1}^n$ were observed for all objects.

Throughout this subsection, the source catalog $\{(s_i,\widetilde Y_i)\}_{i=1}^n$, the true labels $\{Y_i\}_{i=1}^n$, and the random catalog are treated as fixed. Expectations are taken only with respect to the random sampling of the labeled subset $L$, assumed to be a simple random sample of size $m$ drawn without replacement from $\{1,\dots,n\}$.

\paragraph*{Oracle estimator.} Define the oracle (true-label) pair-count numerators and denominators
\begin{align}
N_{DD}^\star(b) &:= \sum_{i<j} w_{ij}\,Y_iY_j\,K_b(d_{ij}), &
S_2^\star &:= \sum_{i<j} w_{ij}\,Y_iY_j, \nonumber\\
N_{DR}^\star(b) &:= \sum_{i,a} w^R_{ia}\,Y_i\,K_b(d_{ia}), &
S_{1R}^\star &:= \sum_{i,a} w^R_{ia}\,Y_i . \nonumber
\end{align}
The corresponding normalized histograms are
\begin{align}
DD^\star(b) &:= \frac{N_{DD}^\star(b)}{S_2^\star}, \\
DR^\star(b) &:= \frac{N_{DR}^\star(b)}{S_{1R}^\star},
\end{align}
and the oracle Landy--Szalay estimator is
\begin{equation}
\widehat\xi^\star(b)
=
\frac{DD^\star(b)-2\,DR^\star(b)+RR(b)}{RR(b)}.
\end{equation}

\paragraph*{Design-unbiased recovery of pair counts.}
The PP construction replaces $Y_i$ by $\widetilde Y_i + \Delta_i$, where $\Delta_i = Y_i-\widetilde Y_i$ is observed only on $L$ and set to zero otherwise. Writing
\[
Y_iY_j
=
\widetilde Y_i\widetilde Y_j
+
\Delta_i\widetilde Y_j
+
\widetilde Y_i\Delta_j
+
\Delta_i\Delta_j,
\]
and applying Horvitz--Thompson scaling to the terms involving $\Delta$, yields estimators
$\widehat{DD_\mathrm{PP}}(b)$,
$\widehat S_{2}$,
$\widehat{DR_\mathrm{PP}}(b)$,
and $\widehat S_{1R}$ as defined in Sec.~\ref{sec:setup}.

\begin{lemma}[Unbiasedness]
\label{lem:pp_unbiased_counts}
If $L$ is a simple random subset of size $m$, then
\begin{align}
\mathbb{E}_L\!\left[\widehat{DD_\mathrm{PP}}(b)\right]
&= DD^\star(b), &
\mathbb{E}_L\!\left[\widehat S_{2}\right]
&= S_2^\star, \\
\mathbb{E}_L\!\left[\widehat{DR_\mathrm{PP}}(b)\right]
&= DR^\star(b), &
\mathbb{E}_L\!\left[\widehat S_{1R}\right]
&= S_{1R}^\star .
\end{align}
\end{lemma}

\begin{proof}
Under fixed-size simple random sampling,
$\Pr(i\in L)=m/n$ and
$\Pr(i,j\in L)=m(m-1)/[n(n-1)]$ for $i\neq j$.
Hence, for arbitrary functions $g$ and $h$,
\begin{widetext}
\[
\mathbb{E}_L\!\left[\frac{n}{m}\sum_{i\in L} g(i)\right] = \sum_{i=1}^n g(i), \qquad \mathbb{E}_L\!\left[\frac{n(n-1)}{m(m-1)} \sum_{i\ne j\in L} h(i,j) \right] = \sum_{i\ne j} h(i,j).
\]
\end{widetext}
Applying these identities to the residual terms in the decomposition of $Y_iY_j$ and to the linear $Y_i$ term in $DR$ yields the result.
\end{proof}

Thus, in expectation over the labeled subset, the PP construction recovers the oracle \emph{unnormalized} pair counts exactly.

\paragraph*{Equivalence of PP-LS and oracle LS.}
Define the normalized PP histograms
\begin{equation}
\widetilde{DD}_{\mathrm{PP}}(b) = \frac{\widehat{DD_\mathrm{PP}}(b)}{\widehat S_{2}},
\qquad
\widetilde{DR}_{\mathrm{PP}}(b) = \frac{\widehat{DR_\mathrm{PP}}(b)}{\widehat S_{1R}},
\end{equation}
and
\begin{equation}
\widehat\xi_{\mathrm{PP}}(b)
=
\frac{\widetilde{DD}_{\mathrm{PP}}(b)
-2\,\widetilde{DR}_{\mathrm{PP}}(b)
+RR(b)}{RR(b)}.
\end{equation}

Because ratios of unbiased estimators are not exactly unbiased, we consider the large-sample regime in which denominators are bounded away from zero and pair counts concentrate. By Lemma~\ref{lem:pp_unbiased_counts} and a standard continuous mapping argument,
\begin{align}
\mathbb{E}_L\!\left[\widetilde{DD}_{\mathrm{PP}}(b)\right]
&\approx DD^\star(b), \\
\mathbb{E}_L\!\left[\widetilde{DR}_{\mathrm{PP}}(b)\right]
&\approx DR^\star(b).
\end{align}
Since $RR(b)$ is unchanged,
\begin{equation}
\mathbb{E}_L\!\left[\widehat\xi_{\mathrm{PP}}(b)\right]
\approx
\widehat\xi^\star(b).
\end{equation}

The PP-LS estimator is design-consistent for the oracle LS estimator: in expectation over the labeled subset, it reproduces the clustering statistic that would be obtained using the true labels. Combined with the standard LS result
$\mathbb{E}[\widehat\xi^\star(b)]\approx\xi(b)$
under adequate random catalogs and consistent weighting,
it follows that
\begin{equation}
\mathbb{E}_L\!\left[\widehat\xi_{\mathrm{PP}}(b)\right]
\approx
\xi(b).
\end{equation}

One key observation is that no assumptions are required on label classification calibration (in case of soft labels), error rates, or spatial structure of misclassification beyond the simple random sampling of the labeled subset.

\subsection{Implementation Remarks} \label{sec:implimentation}

Most modern pair-counting engines, such as \emph{Corrfunc} \citep{sinha2020corrfunc} and \emph{TreeCorr} \citep{jarvis2004skewness,2015ascl.soft08007J}, natively support per-object weights, so the proposed estimator can be implemented without modifying the underlying pair-search routines. The key is to encode the decomposition directly via object-level weights, then combine a small number of standard weighted pair counts with appropriate rescaling.

Assign to each object two weights. First, define $u_i := \widetilde{Y}_i$ for all objects. Second, define a label-restricted weight
\begin{equation}
v_i :=
\begin{cases}
\Delta_i, & i \in L,\\
0, & \text{otherwise},
\end{cases}
\end{equation}
so that $v_i$ is nonzero only for labeled objects. With these two weight fields in place, all required terms can be constructed from ordinary weighted data–data and data–random runs.

We implement the PP-LS estimator by encoding the decomposition in two per-object weight fields $u_i$ and $v_i$, then computing a small set of weighted pair-counts and applying finite-sample scaling factors; the full step-by-step procedure is given in Algorithm~\ref{alg:pp_estimator}. For each separation bin $b$, compute the following ordered weighted pair counts. The term $T_0$ is obtained from a standard data–data run over the full sample with weights $u \times u$. The cross terms are computed by restricting one leg of the ordered pair to labeled objects: $T_1$ uses weights $v \times u$ (labeled $\to$ all), and $T_1'$ uses weights $u \times v$ (all $\to$ labeled). Each of these cross terms is multiplied by the finite-sample correction factor $n/m$. The purely labeled contribution $T_2$ is formed from labeled–labeled ordered pairs with weights $v \times v$, scaled by $n(n-1)/[m(m-1)]$. Summing these ordered contributions and dividing by two yields the symmetrized estimator $\widehat{DD}_{\mathrm{PP}}$.

The DR term can be computed analogously with two runs: a full data–random count using weights $u$, and a labeled-only data–random count using weights $v$, the latter scaled by $n/m$. The random–random term remains unchanged. Throughout, the same binning kernel $K_b$ must be used consistently for $DD$, $DR$, and $RR$ to preserve unbiasedness and ensure that the estimator matches its theoretical definition.

\begin{table}[th!]
\captionsetup{name=Algorithm}
\caption{Compute pair-count building blocks for the PP estimator.}
\label{alg:pp_estimator}
\begin{ruledtabular}
\begin{minipage}{0.95\textwidth}
\begin{algorithmic}[1]
\Require $\{\widetilde Y_i,\Delta_i\}$; $L$; random catalog; kernel $K_b$
\Ensure $\widehat{DD}_{\mathrm{PP}}(b)$, $\widehat{DR}(b)$, $\widehat{RR}(b)$ for each bin $b$

\For{each object $i$}
    \State $u_i \gets \widetilde Y_i$
    \If{$i \in L$}
        \State $v_i \gets \Delta_i$
    \Else
        \State $v_i \gets 0$
    \EndIf
\EndFor

\State Compute ordered data--data weighted sums per bin $b$:
\State $T_0^{\mathrm{ord}}(b) \gets \sum u_i u_j$
\State $T_1^{\mathrm{ord}}(b) \gets \sum v_i u_j$
\State $T_1'^{\mathrm{ord}}(b) \gets \sum u_i v_j$
\State $T_2^{\mathrm{ord}}(b) \gets \sum v_i v_j$

\State $\widetilde T_1(b) \gets \frac{n}{m} T_1^{\mathrm{ord}}(b)$
\State $\widetilde T_1'(b) \gets \frac{n}{m} T_1'^{\mathrm{ord}}(b)$
\State $\widetilde T_2(b) \gets \frac{n(n-1)}{m(m-1)} T_2^{\mathrm{ord}}(b)$

\State $\widehat{DD}_{\mathrm{PP}}(b) \gets \frac{1}{2}\bigl[
T_0^{\mathrm{ord}}(b)
+ \widetilde T_1(b)
+ \widetilde T_1'(b)
+ \widetilde T_2(b)
\bigr]$

\State Compute $DR_u(b)$ with weights $u$
\State Compute $DR_v(b)$ with weights $v$
\State $\widehat{DR}(b) \gets DR_u(b) + \frac{n}{m} DR_v(b)$

\State Compute $\widehat{RR}(b)$ with standard random--random pairs
\end{algorithmic}
\end{minipage}
\end{ruledtabular}
\end{table}

Asymptotic normality of PP-LS estimator follows from classical U-statistic theory applied to the symmetric pair kernel, together with a first-order linearization of the Horvitz--Thompson correction term. Under mild regularity conditions on the sampling design and bounded second moments of the weights, the estimator admits a H\'ajek projection whose leading term is Gaussian. Consequently, the estimator is asymptotically normal within each separation bin and jointly asymptotically normal across bins.

For practical variance estimation, we suggest a spatial resampling scheme that respects the survey geometry and long-range clustering. A spatial jackknife or block bootstrap provides a robust, design-consistent estimate of the full covariance matrix. The survey footprint is partitioned into $B$ contiguous regions of approximately equal area (or equal effective object count). For each leave-one-out region, the estimator must be recomputed in its entirety, including all components
$(\widehat T_0,\quad
\widehat T_1,\quad
\widehat T_1',\quad
\widehat T_2,\quad
\widehat{DR}_{\mathrm{PP}},\quad
\widehat{RR})$.
The jackknife covariance is then formed from the resulting pseudo-values across bins in the standard manner. Because each subsample recomputes every term, this procedure captures cosmic variance, survey mask effects, and variability arising from the residual-based correction.

The magnitude of the covariance depends on the calibration quality of the soft labels $\widetilde Y$. When these are well calibrated \citep{vashistha2025mathcali}, the Horvitz--Thompson correction terms are typically small, leading to reduced variance. In contrast, when hard labels are used, the correction terms would be larger, increasing sensitivity to labeled-sample fluctuations and inflating the covariance.

Finally, although our estimator permits per-object weighting, we do not recommend using ad hoc weights to remove systematics or to correct biases arising from misclassification and contamination \citep{ross2012clustering,elvin2018dark}. The proposed construction explicitly accounts for these effects through its correction terms, rendering additional debiasing weights unnecessary. Introducing external weighting schemes to address contamination or incompleteness risks double-correcting the signal or obscuring the interpretation of the estimator. That said, weights may still be employed for variance reduction, for example, to optimize signal-to-noise under heteroskedastic measurement uncertainty, provided they are not used as substitutes for systematic-bias correction.

\section{Benchmarks} \label{sec:benchamrk}

We compare the proposed estimator with the \emph{cross-correlation decontamination} (CCD) estimator. CCD decontaminates clustering statistics by leveraging (i) a nearly pure contaminant catalog and (ii) sample purity estimates obtained from a labeled subset; see Appendix~\ref{app:ccd-derivation} for a derivation. We follow the common LS within CCD, where pairs between the science sample and a pure contaminant catalog are subtracted using purity weights. Besides performance, the practical advantage of our estimator is that it does \emph{not} require access to a pure contaminant catalog, which is often infeasible, and it relaxes several constraining assumptions about purity estimation.

We evaluate all methods on controlled synthetic skies generated via a Thomas process and a modular \emph{Synthetic Field Construction} pipeline (Appendix~\ref{app:Thomas_Process}). The pipeline exposes knobs for (a) intrinsic clustering of the target (galaxy) population, (b) spatial structure of contaminants, and (c) supervision conditions (labeled subset size, label noise, and calibration).

\paragraph*{Synthetic sky generation.}

We simulate square survey windows $W = [0,L]^2$ with $L=1$ (dimensionless units), using periodic boundary conditions for pair counting in order to eliminate edge effects and isolate estimator behavior. The \emph{true} galaxy field is generated from a Thomas process (see Appendix~\ref{app:Thomas_Process}), which produces a positively correlated clustered point pattern with analytically controlled small-scale structure (see Figure~\ref{fig:toy_example} as an example).

The choice of a Thomas process \citep{thomas1949generalization}, which is similar to Neyman--Scott's model \citep{neyman1958statistical,daley2003introduction}, rather than an $N$-body simulation with a fully realistic survey pipeline, is deliberate. Our objective is not to reproduce the detailed phenomenology of large-scale structure, but to evaluate the statistical properties (bias and variance) of competing estimators under controlled and repeatable conditions. The Thomas process provides (i) explicit control over clustering strength through the parent intensity, offspring distribution, and dispersion scale; (ii) stationarity and isotropy, which simplify interpretation; and (iii) computational efficiency that permits thousands of independent realizations for precise Monte Carlo assessment and variance estimation. In contrast, $N$-body simulations introduce additional layers of modeling complexity (halo finding, galaxy-halo connection, survey masks, selection functions) that are unnecessary for isolating label-misclassification effects and would substantially reduce the number of realizations achievable for variance estimation.

\paragraph*{Noisy labels and scores.}

Each detected object $i$ receives a classifier score $s_i\in[0,1]$. In practice, this score can be estimated from simple features (photo-$z$, magnitude, color, and a context feature approximating local number density). Conditional on class, $s_i$ follows a parametric location-scale family with partial overlap:
\begin{equation*}
s_i\mid Y_i=1 \sim f_1(\cdot;\theta_1),\qquad
s_i\mid Y_i=0 \sim f_0(\cdot;\theta_0),    
\end{equation*}
with parameters selected to achieve the desired operating points. The observed \emph{noisy hard label} $\widetilde Y_i\in \{0,1\}$ is obtained by thresholding at $\tau$ and flipping with false-positive/negative rates $(\alpha,\beta)$ to emulate annotation noise in a labeled subset. In practice, a small labeled set $C$, where $L \cap C=\varnothing$, may be reserved for calibrating $s_i\mapsto \widehat p_i\approx \Pr(Y_i = 1 \mid s_i)$ using isotonic or Platt calibration \citep{platt1999probabilistic,zadrozny2002transforming,guo2017calibration}; the large unlabeled set $U = S\backslash L$ is used for clustering estimation. 
See Figure~\ref{fig:contamination_error_fields} for the expected false positive and false negative rates. This setup deliberately violates the assumption of spatially independent label noise, providing a stringent test of the robustness of the PP-LS estimator to realistic, spatially correlated classification systematics.

\paragraph*{Benchmark estimators.}
We compare the proposed PP-LS against several representative alternatives:

\begin{itemize}
\item \textbf{CCD (Cross-Correlation Decontamination).} 
A contamination-correction approach based on cross-correlating the observed sample with a (nearly) pure contaminant catalog. This method requires (i) access to a high-purity contaminant sample and (ii) an external estimate of the target-sample purity $\widehat{\pi}$, typically obtained from the labeled subset $L$ (see Appendix~\ref{app:ccd-derivation}). Its performance is sensitive to the quality of the contaminant template and the accuracy of the purity estimate.

\item \textbf{Naïve thresholding.} 
The standard LS estimator applied directly to the observed inclusion set $\widetilde G = \{s_i : \widetilde Y_i = 1\}$, ignoring label noise. This estimator targets the clustering of the noisy catalog rather than the truth-level population $G$, and is generally biased when $\widetilde Y_i \neq Y_i$.

\item \textbf{Spectroscopic-only sample.} 
The LS estimator applied exclusively to the labeled subset $L$. This estimator is unbiased for the clustering of $G$ under the simple-random-sampling assumption, but discards the vast majority of objects and therefore represents a worst-case scenario in terms of statistical efficiency.

\item \textbf{Oracle.} 
The LS estimator computed on the unobserved true catalog $G = \{s_i : Y_i = 1\}$. This estimator is infeasible in practice and serves as a lower bound on achievable variance.
\end{itemize}

In contrast to these baselines, PP-LS is applied to the full source catalog $S$ while incorporating the high-fidelity information contained in $L$. It is unbiased for the truth-level clustering under the stated design assumptions and retains the statistical efficiency of the large photometric sample. We also note that several decontamination methods based on cross-correlation with a contaminated sample have been proposed \citep{pullen2016interloper, awan2020angular, foroozan2022correcting}. In this work, we adopt the approach of \citep{foroozan2022correcting}, which performs well for our specific setup.

\subsection{Evaluation Protocol}\label{subsec:protocol}

For each experiment, we generate $N_{R}$ random realizations.  Distances are binned into 15 logarithmic bins from $r_{\min}=0.01L$ to $r_{\max}=0.2L$ with a \emph{top-hat} kernel $K_b$ (see Sec.~\ref{sec:kernel_def}). Random catalogs use $n_R=8000$ points. 

\begin{figure*}
    \centering
    \includegraphics[width=1\linewidth]{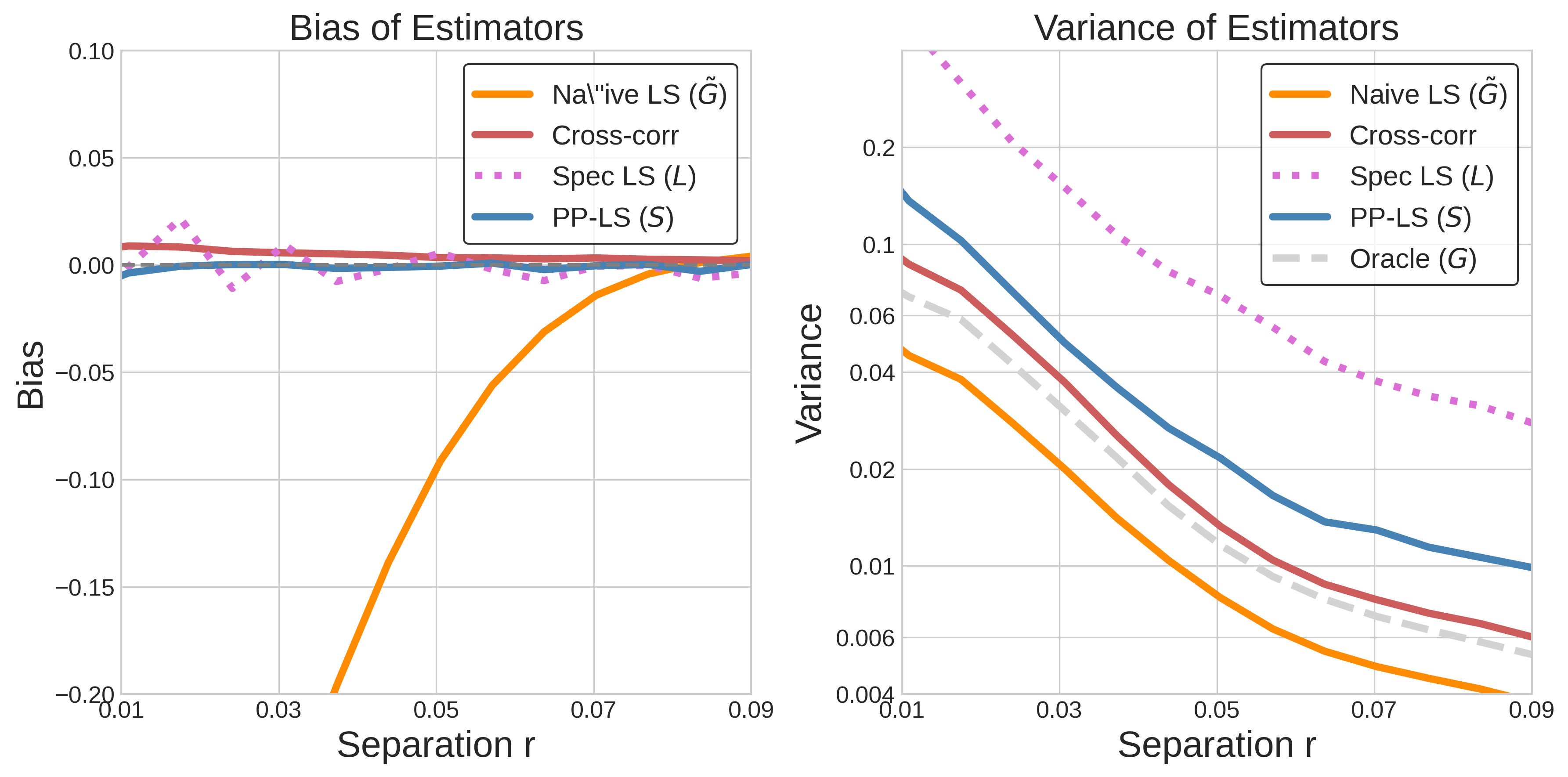}
    \caption{Bias (left) and variance (right) of 2PCF estimators as a function of separation $r$, averaged over $N_R=2000$ simulated realizations. Bias is computed relative to the oracle estimator on sample $G$, which has access to the true labels. The Na\"ive LS estimator exhibits strong small-scale bias but low variance, while the cross-correlation and PP-LS estimators remain approximately unbiased with increased variance. The LS estimator on the spectroscopic sample $L$ shows the largest variance due to reduced sample size ($m/n \ll 1$). The oracle estimator provides the minimum-variance reference for unbiased estimators.}
    \label{fig:bias_variance}
\end{figure*}

{\bf Metrics.} Let $\xi^\star(r)$ denote the known oracle two-point function estimated from the true labels, assuming all objects of interest are observed and correctly classified. We report
\begin{align*}
\text{Bias}(r)&=\mathbb{E}[\widehat{\xi}(r)]-\xi^\star(r),\quad {\rm and} \\
\text{Var}(r) &= \mathbb{E}\left[\big(\widehat{\xi}(r) - \mathbb E[\widehat{\xi}(r)]\big)^2\right].
\end{align*}
However, our focus is on variance, since all estimators, except Na\"ive photo-$z$ LS estimator, considered here are unbiased. Throughout, variances are reported relative to the oracle variance, where the oracle estimator assumes that the true labels of all sources are perfectly known -- corresponding to an idealized, fully spectroscopic survey with no label uncertainty.

\section{Results}\label{subsec:results}

\subsection{Illustrative Example}

Figure~\ref{fig:toy_example} presents a controlled two-dimensional experiment designed to visualize the distinction between the truth-level population $G=\{s_i:Y_i=1\}$ (left panel) and the observed inclusion set  $\widetilde G=\{s_i:\widetilde Y_i=1\}$ (middle panel), and to compare the resulting correlation estimates.

We simulate a field on the unit square $\mathcal{S}=[0,1]^2$ with separations binned over $r\in[0.005,0.25]$ using 15 linear bins. A uniform random catalog of 8,000 points models the survey window. The truth-level target population $G$ is generated from a clustered Thomas process (parent intensity $\lambda_{\rm parent} = 200$, mean offspring number $\langle N_{\rm off} \rangle=25$, dispersion $\sigma=0.03$), producing realistic small-scale clustering. An additional contaminant population is superimposed, drawn from an inhomogeneous process combining a large-scale gradient with a localized hotspot (Figure~\ref{fig:contamination_error_fields}). The detected source catalog $S$ is the union of target objects and contaminants.

Noisy catalog membership is constructed by assigning a binary inclusion label $\widetilde Y_i\in\{0,1\}$ with spatially varying false-positive and false-negative rates. Misclassification rates are elevated near the contaminant hotspot and within a shallow-depth stripe, generating both contamination ($Y_i=0$, $\widetilde Y_i=1$) and incompleteness ($Y_i=1$, $\widetilde Y_i=0$). The resulting observed set $\widetilde G$ has an overall contamination fraction of approximately $30\%$. A labeled subset $L$, comprising $10\%$ of the detected objects and drawn uniformly at random from $S$, provides gold-standard labels on a small fraction of the sample.

The standard LS estimator applied to $\widetilde G$ measures the clustering of the noisy catalog rather than that of $G$ (red dotted curve). In the figure, this estimator is biased high at small separations due to the clustered hotspot contaminants, and biased on large scales due to the imposed gradient. The LS estimator applied to $G$ (oracle) recovers the true clustering by construction and serves as a reference (gray curve). The PP-LS estimator, computed on the full source catalog $S$ with residual corrections estimated only from the labeled subset $L$, closely tracks the oracle curve across all scales (blue dashed curve).

The comparison shows that spatially structured discrepancies between $G$ and $\widetilde G$ can induce scale-dependent bias in standard analyses, and that residual-based correction using a small random labeled subset is sufficient to recover the truth-level two-point correlation function in this setting. Next, we study bias and variance of different estimator choices and configurations.

\subsection{Benchmarking Bias and Variance of 2PCF Estimators}

Figure~\ref{fig:bias_variance} shows the bias (left panel) and variance (right panel) of several estimators of the 2PCF as a function of separation $r$, averaged over $N_R = 2000$ independent realizations. The radial range $r \in [0.001, 0.1]$ is divided into 15 top-hat bins. Bias is defined relative to the oracle estimator, which has access to the true class labels and therefore provides a reference baseline for the underlying clustered signal. Each realization consists of a clustered signal component generated from a Thomas process with parent intensity $\lambda_{\rm parent}=100$, mean offspring number $\langle N_{\rm off} \rangle = 20$, and Gaussian spatial dispersion $\sigma = 0.02$. The signal is embedded in a spatially inhomogeneous contaminant field with expected contamination fraction $30\%$. True labels $Y \in \{0,1\}$ distinguish signal and contaminants. Observed labels $\widetilde Y$ are obtained by introducing symmetric misclassification at an expected rate of $10\%$, uniformly across the field. 

A spectroscopic subsample is drawn uniformly at random from the full catalog, with labeled fraction $|L|/n \approx 0.1$, defining the gold-standard subset used by the cross-correlation and PP-LS estimators. For each realization, we compute five estimators: (i, Oracle) the oracle estimator using true labels for all objects, (ii, Na\"ive LS) the na\"ive estimator using noisy labels in $\widetilde G$, (iii, Cross-corr) a cross-correlation-based decontamination estimator, (iv, PP-LS) the proposed PP-LS estimator on the entire source catalog $S$, and (v, Spec LS) a spectroscopic-only estimator using only objects in the labeled subset $L$. All pair counts are evaluated using a common random catalog of $8000$ uniformly distributed points and the LS normalization. For each separation bin, the bias is computed as the difference between the mean estimator value and the mean oracle estimate across realizations, and the variance is estimated empirically from the same ensemble. 

The left panel displays the estimator bias. The LS estimator on noisy labels exhibits a pronounced negative bias on small scales, reaching $\mathcal{O}(10^{-1})$ at the smallest separations considered. This bias decreases monotonically with increasing $r$ and becomes negligible only at large separations. In contrast, the cross-correlation, PP-LS, and spectroscopic LS estimators are consistent with zero bias across all separations, with residual fluctuations attributable to Monte Carlo noise. The spectroscopic estimator shows slightly larger small-scale scatter, reflecting the reduced effective sample size after label masking, but as anticipated, no systematic bias is observed.

The right panel shows the variance of each estimator. The oracle estimator, while unattainable in practice, serves as a useful reference and shows the minimum achievable variance for unbiased estimators. The LS estimator on noisy labels in $\widetilde G$ exhibits even smaller variance, reflecting its substantially larger effective sample size, though this apparent advantage is offset by its strong bias. The CCD incurs a modest variance penalty relative to the oracle estimator, while the PP-LS estimator shows slightly higher variance, particularly at larger separations where the true correlation function approaches zero. The spectroscopic estimator has the largest variance across all scales due to incomplete labeling and the resulting reduction in sample size. For all estimators, the variance decreases with increasing separation, consistent with the growth in the number of contributing pairs.

When CCD's strong assumptions truly hold, depending on the size of the spectroscopic sample, the CCD method can be more competitive. In the more realistic regimes where the contaminant catalog is not pure, or purity is not known precisely, our estimator is consistently more robust, offering lower small-scale bias and better uncertainty calibration, without requiring an external contaminant catalog.

\begin{figure}
    \centering
    \includegraphics[width=0.98\linewidth]{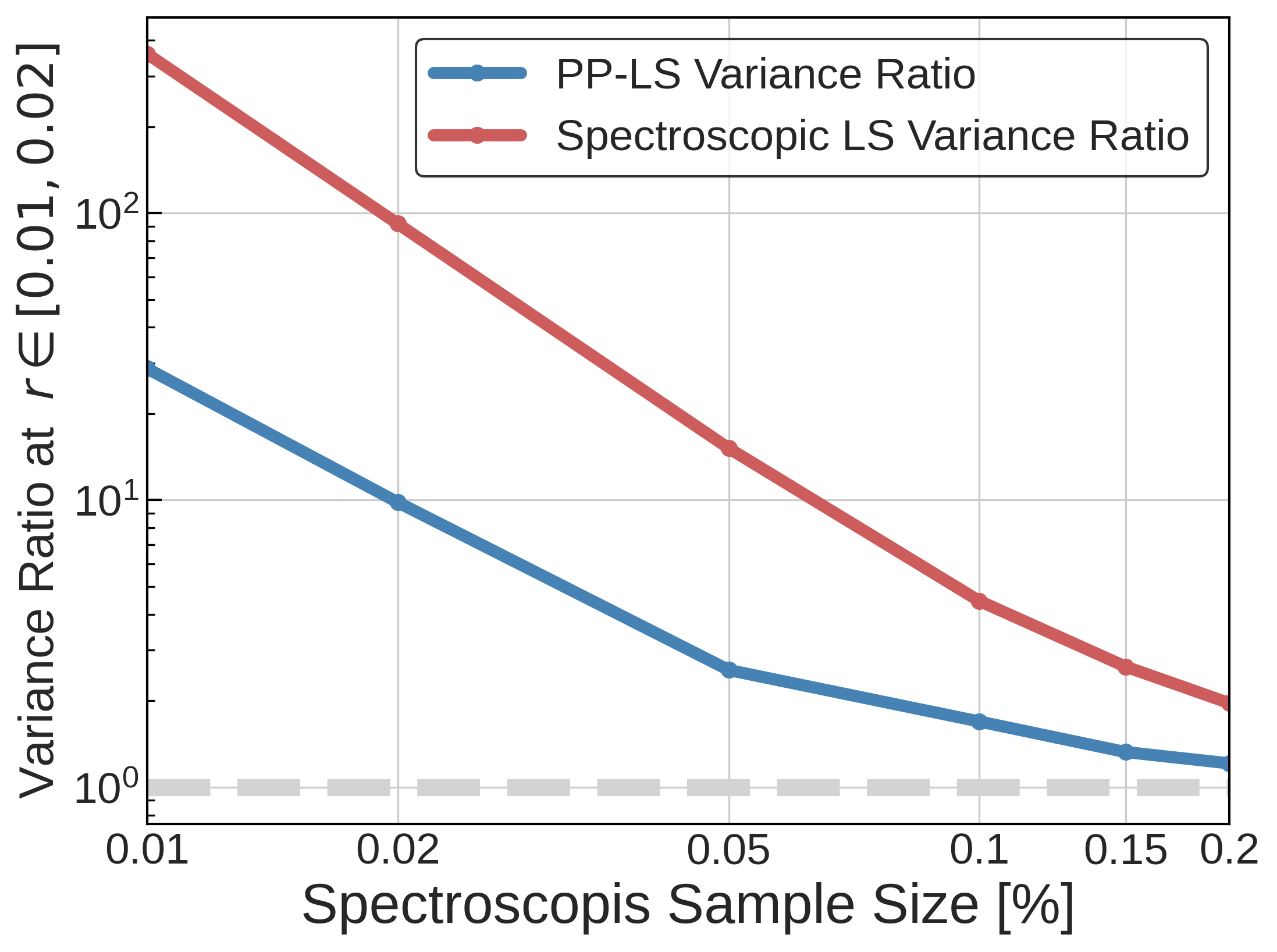}
    \caption{Variance ratio of the PP-LS and LS estimators on $S$ and $L$, relatively, relative to the oracle estimator on $G$ as a function of the labeled sample size, shown for a fixed separation bin $r \in [0.01, 0.02]$. Each point is estimated from 800 independent realizations. The spectroscopic LS estimator exhibits extremely large variance at small labeled fractions, while the PP-LS estimator achieves substantially lower variance by leveraging the unlabeled sample. The dashed horizontal line indicates the oracle baseline variance (normalized to 1).}
\label{fig:spec_sample_variance}
\end{figure}

\subsection{Variance Scaling with Labeled Sample Size}

The PP-LS method is intended for settings in which a small spectroscopic sample is combined with a much larger photometric sample. It is therefore important to quantify the variance reduction achieved relative to using the spectroscopic sample alone, and to characterize how this improvement depends on the relative sizes of the labeled ($L$) and unlabeled ($S \backslash L$) samples. Figure~\ref{fig:spec_sample_variance} shows the variance of the PP-LS and spectroscopic LS estimators as a function of the labeled (spectroscopic) sample size, evaluated at a fixed small-scale separation bin $r \in [0.01, 0.02]$. For each labeled fraction, the variance is estimated from $N_R=800$ independent realizations and normalized by the variance of the oracle estimator for the same bin.

Both estimators exhibit a strong decrease in variance as the labeled fraction increases, reflecting the increasing amount of supervised information contained in the spectroscopic sample. However, their dependence on the labeled sample size $m$ differs substantially. The spectroscopic LS estimator shows extremely large variance when the labeled fraction is small ($\le 5\%$), exceeding the oracle variance by more than one order of magnitude at the lowest labeled fractions considered. Its variance decreases steeply with increasing labeled fraction, approaching the oracle level only when a substantial portion of the sample is labeled ($20\%$).

In contrast, the PP-LS estimator is markedly more efficient in the low-label regime. Even when only $\mathcal{O}(1\%)$ of objects are labeled, the PP-LS variance exceeds the oracle variance by only one to two orders of magnitude, representing an order-of-magnitude improvement over the spectroscopic LS estimator. As the labeled fraction increases, the PP-LS variance rapidly approaches the oracle limit and remains within a factor of a few across the full range shown.

These results demonstrate that PP-LS effectively leverages unlabeled noisy samples to reduce variance at small scales, yielding substantial gains over purely spectroscopic estimation when labels are sparse. This behavior is consistent with interpreting PP-LS as a semi-supervised estimator that smoothly interpolates between the fully supervised ($G$) and label-limited regimes ($L$).

 \subsection{Effect of Source Inclusion Label Error}

Figure~\ref{fig:classification_error} shows the impact of source inclusion label error in the labeled (spectroscopic) sample on the variance of the PP-LS estimator, evaluated at a fixed small-scale separation bin $r \in [0.01, 0.02]$. The inclusion label error rate parameterizes the fraction of mislabeled objects in sample $\widetilde G$ (both false positives and false negatives), while the spectroscopic labeled fraction is held fixed at $2\%$. For each error rate, the variance is estimated from $N_R = 800$ independent realizations and normalized by the variance of the oracle estimator for the same bin.

The simulations are based on clustered signal points generated from a Thomas process with parent intensity $\lambda_{\rm parent}=60$, mean offspring number $\langle N_{\rm off} \rangle = 20$, and spatial dispersion $\sigma=0.05$, embedded in an inhomogeneous contaminant field with an expected contamination fraction of $30\%$. True inclusion labels $Y$ distinguish signal and contaminants perfectly, while observed inclusion labels $\widetilde Y$ are generated by introducing controlled misclassification error: inclusion labels are flipped to contaminants for target sources, and vice versa, with probability given by the classification error rate uniformly across the survey area. A spectroscopic subsample is then drawn uniformly at random from the full catalog, defining the labeled set used by PP-LS.

As the classification error rate increases, the variance of the PP-LS estimator grows monotonically. In the absence of classification error, PP-LS attains near-oracle performance, with variance only slightly above the oracle limit. However, even modest levels of misclassification lead to slow degradation, with the variance increasing by more than an order of magnitude once the error rate exceeds $\sim 30\%$. At high error rates, the PP-LS variance approaches the variance of a purely spectroscopic LS estimator with the same labeled fraction, indicated by the upper dashed line.

This behavior highlights the trade-off between the number of labels (Figure~\ref{fig:spec_sample_variance}) and classification quality (Figure~\ref{fig:classification_error}). While PP-LS can efficiently leverage a small labeled sample to suppress variance by incorporating the full photometric dataset, misclassified labels introduce additional noise that propagates through the estimator. Even so, across the full range of error rates considered, PP-LS consistently outperforms the spectroscopic-only baseline, achieving more than an order-of-magnitude reduction in variance for moderate levels of classification error $\lesssim 20\%$. This behavior demonstrates that PP-LS's robustness to relatively large classification error, as expected in realistic large-scale surveys.

\begin{figure}
    \centering
    \includegraphics[width=0.98\linewidth]{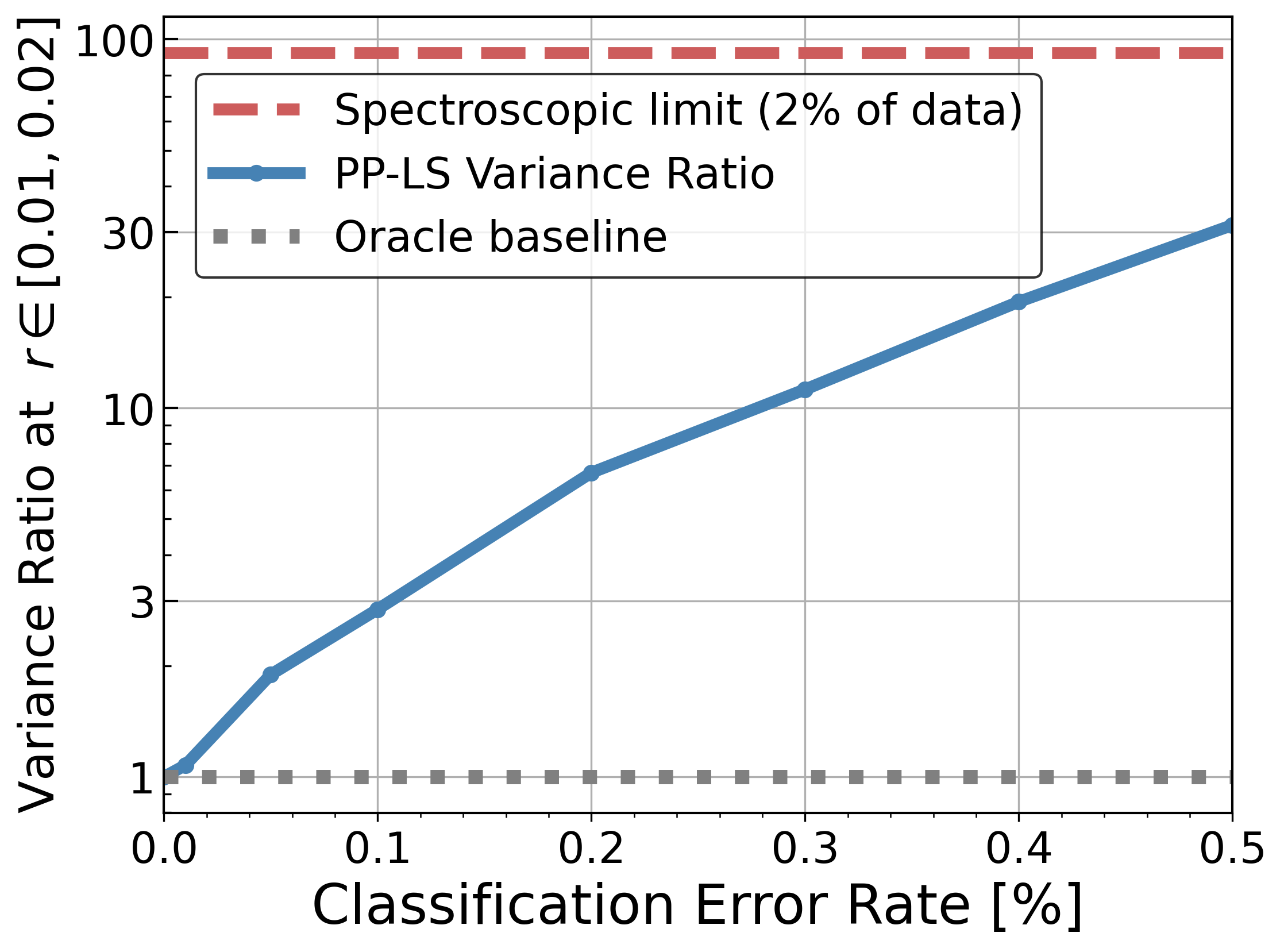}
    \caption{
    Variance ratio of the PP-LS estimator relative to the oracle estimator as a function of the classification error rate in the labeled (spectroscopic) sample, shown for a fixed separation bin $ r \in [0.01, 0.02] $. The labeled fraction is fixed at $2\%$. Target sources are drawn from a Thomas process with a clustered structure and embedded in an inhomogeneous contaminant field. Each point is estimated from $800$ independent realizations. The lower dashed horizontal line indicates the oracle variance, while the upper dashed line shows the variance of the LS estimator on a purely spectroscopic sample $L$ with the $2\%$ labeled fraction.}
    \label{fig:classification_error}
\end{figure}

\subsection{Limitations and Survey Dependence}
\label{sec:limitations}

The simulations presented in this work are designed to isolate estimator behavior under controlled conditions rather than to reproduce the full complexity of a specific cosmological survey. The underlying galaxy field is modeled as a stationary, clustered point process with periodic boundary conditions, and classification errors are introduced in a spatially inhomogeneous, statistically independent manner. Real surveys, however, exhibit additional structure: survey masks, depth variations, spatially varying selection functions, redshift evolution, scale-dependent bias, and nontrivial correlations between photometric quality and environment \citep{abbott2022dark,to2025dark}. Each of these effects can modify both the clustering signal and the effective misclassification mechanism.

As a consequence, the quantitative gain obtained by combining a large, low-fidelity (contaminated) photometric sample with a smaller, high-fidelity spectroscopic subsample is survey dependent. While PP-LS is unbiased (see Sec.~\ref{subsec:theory} for requirements), the relative improvement in variance depends on factors such as the intrinsic clustering amplitude, number density, contamination fraction, the structure of the confusion matrix, spatial variation in classification performance, and the labeled fraction $|L|/n$. In particular, the optimal allocation of photometric and spectroscopic effort is coupled to the survey's depth, footprint, and target-selection strategy. The results reported here, therefore, establish generic trends and scaling behavior, demonstrating that hybrid estimators can reduce bias while retaining a substantial fraction of the statistical power of the full sample. However, precise numerical gains should not be interpreted as universal. A fully realistic assessment requires embedding the estimator within an end-to-end survey model, and the joint optimization of low- and high-fidelity samples is most naturally viewed as a component of survey design.

\section{conclusion} \label{sec:conclusion}

We have introduced a prediction-powered Landy--Szalay (PP-LS) estimator for two-point correlation analysis under noisy source catalog. The method explicitly distinguishes between the truth-level target population 
$G = \{s_i : Y_i=1\}$ and the observed inclusion set  $\widetilde G = \{s_i : \widetilde Y_i=1\}$,
and corrects the resulting discrepancy using a small, high-fidelity labeled subset drawn as a simple random sample from the source catalog. By embedding Horvitz--Thompson–scaled residual corrections directly into the pair counts, PP-LS reconstructs, in expectation over the labeling design, the same clustering statistic that would be obtained if the true labels were observed for all objects.

The construction operates within the standard Landy--Szalay framework: a realistic random catalog is still required to model survey geometry and baseline selection, while the prediction-powered component corrects for misidentification of inclusion (contamination and incompleteness) without requiring calibration of uncertainty in inclusion labels, knowledge of false-positive or false-negative rates, or a parametric model of the contaminant population. Under simple random sampling of the labeled subset, the estimator is design-unbiased and consistent for the truth-level 2PCF, and it preserves the geometric and edge-correction advantages of the classical LS estimator.

Our simulations show that PP-LS remains reliable and unbiased in regimes with spatially structured label errors, including clustered contaminants and redshift misassignment, where na\"ive LS estimators on contaminated samples are biased and spectroscopic-only analyses suffer severe variance inflation due to limited sample size. By leveraging the full source catalog $S$ for pair counting and using the labeled subset $L$ only for bias correction, PP-LS achieves substantial efficiency gains in the low-label regime. It effectively interpolates between the spectroscopic limit (small but unbiased) and the oracle limit (full truth-level information), while requiring only a small number of additional weighted pair counts that can be computed using existing correlation-function pipelines.

This framework opens new avenues for working with contaminated and noisy data. Although photometric redshift contamination is the dominant source of contamination in imaging surveys, the framework is fully general: any mechanism that induces a mismatch between $Y_i$ and $\widetilde Y_i$ (including luminosity scatter, classification error, or spatially correlated systematics) fits naturally into the same formalism. Extensions to unequal-probability or stratified labeling designs, incorporation into tomographic and redshift-space analyses, and generalization to higher-order statistics represent natural directions for future work. As forthcoming surveys such as LSST, \textit{Euclid}, and Roman deliver massive photometric samples with comparatively limited spectroscopy, prediction-powered correlation estimation provides a scalable, design-based, and statistically principled foundation for unbiased clustering inference under realistic observational conditions.

\begin{acknowledgments}
We thank Karl Gebhardt for the fruitful discussions during the early phases of this work and Dragan Huterer for helpful comments. We acknowledge support from the National Science Foundation under Cooperative Agreement 2421782 and the Simons Foundation award MPS-AI-00010515.
\end{acknowledgments}

\newpage
\appendix

\section{Thomas Process and Synthetic Field Construction}
\label{app:Thomas_Process}

The Thomas process is a widely used example of a clustered spatial point process, belonging to the broader class of Neyman--Scott cluster processes. It is defined on a bounded observation window $D \subset \mathbb{R}^2$ (here we take $D=[0,1]^2$) and is parameterized by three quantities: the parent intensity $\lambda_{\rm parent} > 0$, the mean number of offspring per parent $\langle N_{\rm off} \rangle > 0$, and the offspring dispersion parameter $\sigma > 0$.

The process is constructed hierarchically:
\begin{enumerate}
    \item \textbf{Parent generation.} First, a set of ``parent'' points $\{P_i\}_{i=1}^{N_p}$ is generated from a homogeneous Poisson process on $D$ with intensity $\lambda_{\rm parent}$. Thus, the number of parents satisfies
    $
        N_p \sim \text{Poisson}(\lambda_{\rm parent} |D|),
    $
    and conditional on $N_p$, the parent locations $P_i \in D$ are i.i.d.\ uniform over $D$.
    \item \textbf{Offspring generation.} Each parent $P_i$ generates a Poisson-distributed number of offspring,
    $
        N_i \sim \text{Poisson}(\langle N_{\rm off} \rangle),
    $
    independently across parents.
    \item \textbf{Offspring displacement.} Conditional on a parent $P_i = (p_x, p_y)$, the $N_i$ offspring are distributed according to an isotropic bivariate Gaussian density centered at $P_i$ with variance $\sigma^2$ per coordinate:
    $
        (X_{ij}, Y_{ij}) \;\sim\; \mathcal{N}(p_x, \sigma^2) \times \mathcal{N}(p_y, \sigma^2), \quad j=1,\dots,N_i.
    $
    Offspring lying outside $D$ are discarded.
\end{enumerate}

The resulting process $\mathcal{T} = \{(X_{ij},Y_{ij})\}$ is a clustered point pattern: parent points form a homogeneous Poisson process, and offspring are Gaussian-distributed around them. The degree of clustering is controlled by $\langle N_{\rm off} \rangle$ (average cluster size) and $\sigma$ (spatial spread). When $\langle N_{\rm off} \rangle=1$ and $\sigma \to 0$, the process converges to a homogeneous Poisson process.

\subsection*{Synthetic Field with Contaminants and Label Noise}

To construct the synthetic field used in our experiments, we superimpose two spatial components:
\begin{enumerate}
    \item \textbf{Galaxy population.} True galaxies are generated via the Thomas process with parameters $(\lambda_{\rm parent}, \langle N_{\rm off} \rangle, \sigma)$, producing a clustered distribution of points.
    \item \textbf{Contaminant population.} Contaminants are drawn from an inhomogeneous Poisson process with spatially varying intensity function
    $
        \lambda_c(x,y) = 0.5 + 1.5x + 2 \exp \left[-\frac{(x-0.75)^2 + (y-0.25)^2}{2 \times 0.03^2}\right].
    $
    This intensity consists of a smooth linear gradient in $x$ together with a localized ``hotspot'' centered at $(0.75,0.25)$.
\end{enumerate}

The two populations are combined into a single catalog of positions $\{(x_k,y_k)\}_{k=1}^N$. Each point is assigned a \emph{true class label} $Y_k \in \{0,1\}$, with $Y_k=1$ for galaxies and $Y_k=0$ for contaminants.

\paragraph*{Label noise.} The observed labels $\widetilde{Y}_k$ are subject to spatially varying misclassification. Specifically,
\begin{align}
    \Pr(\widetilde{Y}_k = 0 \mid Y_k = 1) &\;=\; \pi_{\mathrm{FN}}(x_k,y_k), \\
    \Pr(\widetilde{Y}_k = 1 \mid Y_k = 0) &\;=\; \pi_{\mathrm{FP}}(x_k,y_k),
\end{align}
where $\pi_{\mathrm{FN}}$ and $\pi_{\mathrm{FP}}$ encode false negative and false positive rates, respectively. In our construction:
\begin{itemize}
    \item The false negative probability $\pi_{\mathrm{FN}}(x,y)$ is elevated near the lower boundary ($y < 0.2$), reflecting difficulty in classifying galaxies in that region.
    \item The false positive probability $\pi_{\mathrm{FP}}(x,y)$ increases both with $x$ and near the contaminant hotspot at $(0.75,0.25)$, mimicking structured confusion between galaxies and contaminants.
\end{itemize}

\paragraph*{Spectroscopic subset.} Finally, a subset of size $m_{\text{labels}}$ is drawn uniformly at random from the full catalog to represent spectroscopic ground truth, yielding a mask of labeled versus unlabeled points.

The resulting synthetic field has four key features of astronomical survey data: (\emph{i}) galaxies form clustered spatial structures well captured by a Thomas process, (\emph{ii}) contaminants arise from a structured, inhomogeneous background, (\emph{iii}) observed labels are corrupted by spatially dependent misclassification, and (\emph{iv}) only a partial spectroscopic subset is available for supervised learning. This synthetic construction provides a controlled environment for testing classification and decontamination methods under realistic conditions.

\paragraph*{Spatial structure of contamination and classification errors.}
Figure~\ref{fig:contamination_error_fields} illustrates the spatial structure of the contaminant intensity field and the induced patterns of classification error used in our simulations. These fields are designed to introduce realistic, spatially varying systematics that affect both the true source population and the quality of the observed labels. The left panel shows the contaminant intensity field, which combines a smooth large-scale gradient across the survey footprint with a localized overdensity. Contaminant sources are drawn from an inhomogeneous Poisson process whose intensity varies spatially according to this field, producing position-dependent contamination fractions. The middle panel displays the spatial density of false positives, corresponding to contaminant sources that are misclassified as signal. This field closely traces regions of high contaminant intensity, with the highest false-positive rates occurring in the localized overdensity. The right panel shows the spatial density of false negatives, corresponding to true signal sources that are misclassified as contaminants. In contrast to the false positives, the false-negative density exhibits a distinct spatial pattern, with enhanced misclassification concentrated in a band-like region across the field.

\begin{figure*}[th!]
    \centering
    \includegraphics[width=0.98\linewidth]{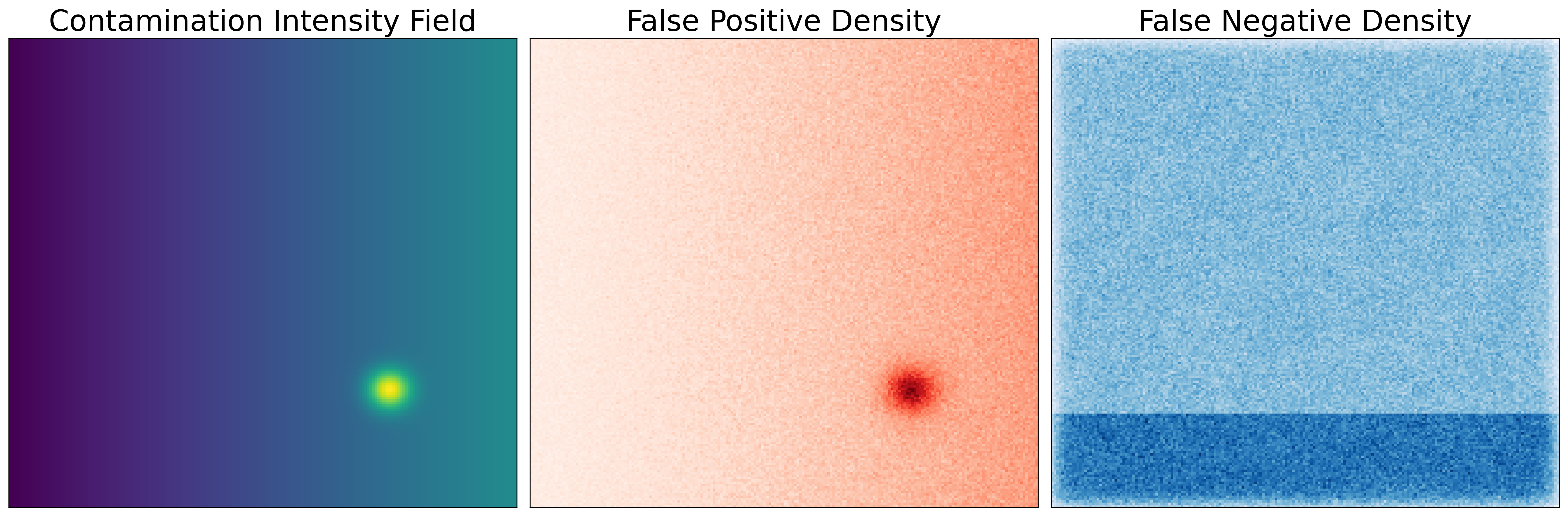}
    \caption{
    Spatial structure of contamination and classification errors used in the simulations. \textit{Left:} contaminant intensity field defining the inhomogeneous Poisson process from which contaminant sources are drawn, including a smooth gradient and a localized overdensity. \textit{Middle:} spatial density of false positives (contaminants misclassified as signal). \textit{Right:} spatial density of false negatives (signal sources misclassified as contaminants). The spatial correlation between these fields introduces realistic, position-dependent classification errors.}
    \label{fig:contamination_error_fields}
\end{figure*}

\section{Cross-Correlation Decontamination Estimator}
\label{app:ccd-derivation}

A central challenge in galaxy surveys is to recover the true clustering signal of galaxies when observed catalogs are contaminated by non-galaxies (e.g.\ stars or spurious detections). The \emph{cross-correlation estimator} provides a method to decontaminate clustering statistics by leveraging (i) a nearly pure contaminant catalog and (ii) estimates of the sample purity from a labeled subset.

Let $A$ denote the observed catalog of objects labeled as galaxies by the noisy classifier. This catalog contains a mixture of true galaxies ($G$) and contaminants ($C$). Following the contamination model of \citep{pullen2016interloper}, we write the observed fraction of true galaxies among $A$ as
$
    a = \frac{|A \cap G|}{|A|},
$
which we call the \emph{purity}. The purity $a$ is not known exactly but can be estimated, denoted $\widehat{a}$, from a spectroscopic subset with ground-truth labels. If known then the decontamination methods such as \cite{awan2020angular} can be applied. 

We are interested in recovering the true galaxy--galaxy correlation function, $\xi_{GG}(r)$, at separation scale $r$. Following \citep{foroozan2022correcting}, the observed two-point functions involve mixtures
\begin{align*}
    \xi_{AA}(r) &= a^2 \xi_{GG}(r) + 2a(1-a)\xi_{GC}(r) + (1-a)^2 \xi_{CC}(r), \\
    \xi_{AC}(r) &= a \, \xi_{GC}(r) + (1-a) \xi_{CC}(r),
\end{align*}
where
\begin{itemize}
    \item $\xi_{AA}(r)$ is the autocorrelation of the observed ``galaxy'' sample $A$,
    \item $\xi_{CC}(r)$ is the autocorrelation of a (near-)pure contaminant catalog $C$,
    \item $\xi_{AC}(r)$ is the cross-correlation between $A$ and the pure contaminant sample,
    \item $\xi_{GC}(r)$ is the unknown galaxy--contaminant cross-correlation.
\end{itemize}

From the second equation, we can isolate $\xi_{GC}(r)$:
\begin{equation}
    \xi_{GC}(r) \;=\; \frac{\xi_{AC}(r) - (1-a) \xi_{CC}(r)}{a}.    
\end{equation}
Substituting this into the first equation yields
\begin{equation}
    \xi_{GG}(r) \;=\; \frac{ \xi_{AA}(r) - 2a(1-a) \xi_{GC}(r) - (1-a)^2 \xi_{CC}(r) }{a^2}.
\end{equation}
Thus, given empirical estimates of $\xi_{AA}$, $\xi_{AC}$, and $\xi_{CC}$, along with an estimate of $a$, we obtain an estimator
\begin{equation}
    \widehat{\xi}_{GG}(r) \;=\; \frac{ \widehat{\xi}_{AA}(r) - 2\widehat{a}(1-\widehat{a}) \widehat{\xi}_{GC}(r) - (1-\widehat{a})^2 \widehat{\xi}_{CC}(r)}{\widehat{a}^2},
\end{equation}
with
\begin{equation}
    \widehat{\xi}_{GC}(r) \;=\; \frac{ \widehat{\xi}_{AC}(r) - (1-\widehat{a})\widehat{\xi}_{CC}(r)}{\widehat{a}}.
\end{equation}

In practice, the correlation functions $\widehat{\xi}_{AA}, \widehat{\xi}_{AC}, \widehat{w}_{CC}$ are estimated using the LS estimator with random catalogs. Purity $\widehat{a}$ is computed from the labeled subset. If labeled objects exist in $A$, then $\widehat{a}$ is simply the fraction of true galaxies among them; otherwise, it is derived from the overlap of labeled true and observed galaxies.

The cross-correlation decontamination method provides a principled way to recover $\xi_{GG}$ in the presence of classification noise and structured contaminants. Its key ingredients are:
\begin{enumerate}
    \item A nearly pure contaminant catalog for estimating $\xi_{CC}$ and $\xi_{AC}$,
    \item A labeled subset for estimating the purity $\widehat{a}$,
    \item Standard correlation estimators (e.g.\ LS) applied consistently across catalogs.
\end{enumerate}
By explicitly solving the mixture equations for $\xi_{GG}$, the method corrects for both contamination fraction and galaxy--contaminant correlations, yielding an unbiased estimator of the true clustering signal.

\section{Derivation of the Landy--Szalay Estimator}\label{app:LS_derivation}

This appendix derives the expectation of the LS estimator in a way that cleanly separates (i) the \emph{geometry/selection} encoded by the survey window and weights from (ii) the \emph{clustering} encoded by the two-point correlation function $\xi$.

\paragraph*{Point process, selection, and random catalog.} Let $x\in\mathcal{V}\subset\mathbb{R}^3$ denote position and let $\varphi(x)\in[0,1]$ denote the (possibly spatially varying) survey window/selection function.
We consider a \emph{galaxy} point process with first-order intensity
\begin{equation}
\lambda(x)=\bar n\,\varphi(x),
\end{equation}
where $\bar n$ is a reference mean density (e.g.\ the mean density before masking).

We assume that, conditional on the window/selection, the process has second-order product density of the usual form
\begin{equation}\label{eq:lambda2}
\lambda^{(2)}(x_1,x_2)
=
\lambda(x_1)\lambda(x_2)\bigl[1+\xi(r_{12})\bigr],\ r_{12}:=\|x_1-x_2\|.
\end{equation}
This is equivalent to Equation~\eqref{eq:xi-prob} in standard treatments: $\xi(r)$ measures departures from Poisson via the pair-product density.

Let $R$ be an independent Poisson process on $\mathcal{V}$ with intensity
\begin{equation}\label{eq:random_intensity}
\lambda_R(x)=\bar n_R\,\varphi_R(x),
\end{equation}
where $\bar n_R$ is an arbitrary constant and $\varphi_R$ is the selection function used to generate randoms.
In the ideal LS construction, $\varphi_R$ is chosen to match the \emph{mean selection of the data sample being correlated}. In the simplest case (no additional selection modulation beyond $\varphi$), one takes $\varphi_R=\varphi$.

Let $K_b(r)$ denote the binning kernel for separation bin $b$ (e.g. a top-hat indicator), and let $W(x_1,x_2)\ge 0$ and $W^R(x_1,x_2)\ge 0$ be symmetric pair-weight functions used for $DD$ and for $DR/RR$, respectively.
Throughout, we assume the same $K_b$ is used for $DD$, $DR$, and $RR$.

Define the \emph{normalized} pair-count histograms by
\begin{align}
DD(b) &:= \frac{\sum_{i<j\in D} W(x_i,x_j)\,K_b(r_{ij})}
{\sum_{i<j\in D} W(x_i,x_j)},
\\[3pt]
DR(b) &:= \frac{\sum_{i\in D}\sum_{a\in R} W^R(x_i,x_a)\,K_b(r_{ia})}{\sum_{i\in D}\sum_{a\in R} W^R(x_i,x_a)},
\\[3pt]
RR(b) &:= \frac{\sum_{a<b\in R} W^R(x_a,x_b)\,K_b(r_{ab})}{\sum_{a<b\in R} W^R(x_a,x_b)}.
\end{align}
The LS estimator is then
\begin{equation}\label{eq:LS_def}
\widehat\xi_{\rm LS}(b) = \frac{DD(b)-2DR(b)+RR(b)}{RR(b)}.
\end{equation}

We now compute the expectations of these normalized histograms in the large-count (or large-random) regime. Throughout, we use Campbell's theorem and ignore the common $1/2$ factor arising from unordered pairs, since it cancels in ratios.

\paragraph*{Expectation of $DD(b)$.} By Campbell's theorem \citep{daley2003introduction} for second-order sums of a point process,
\begin{widetext}
\begin{align}
\mathbb{E}\!\left[\sum_{i<j\in D} W_{ij}K_b(r_{ij})\right]
&=
\frac12\iint \lambda^{(2)}(x_1,x_2)\,W(x_1,x_2)\,K_b(r_{12})\,dx_1dx_2,
\\
\mathbb{E}\!\left[\sum_{i<j\in D} W_{ij}\right]
&=
\frac12\iint \lambda^{(2)}(x_1,x_2)\,W(x_1,x_2)\,dx_1dx_2.
\end{align}
\end{widetext}
Substituting~\eqref{eq:lambda2} and canceling constants gives the exact ratio form
\begin{equation}\label{eq:EDD_exact_ratio}
\mathbb{E}[DD(b)]
=
\frac{\iint \lambda_1\lambda_2\,[1+\xi(r_{12})]\,W_{12}\,K_b(r_{12})\,dx_1dx_2}
{\iint \lambda_1\lambda_2\,[1+\xi(r_{12})]\,W_{12}\,dx_1dx_2},
\end{equation}
where $\lambda_i:=\lambda(x_i)$ and $W_{12}:=W(x_1,x_2)$.

It is often useful to split the numerator into the ``Poisson baseline'' plus a ``clustering excess'' term
\begin{widetext}
\begin{equation}\label{eq:EDD_split}
\mathbb{E}[DD(b)] = \frac{\iint \lambda_1\lambda_2\,W_{12}\,K_b(r_{12})\,dx_1dx_2}
{\iint \lambda_1\lambda_2\,[1+\xi(r_{12})]\,W_{12}\,dx_1dx_2} + \frac{\iint \lambda_1\lambda_2\,\xi(r_{12})\,W_{12}\,K_b(r_{12})\,dx_1dx_2}
{\iint \lambda_1\lambda_2\,[1+\xi(r_{12})]\,W_{12}\,dx_1dx_2}.
\end{equation}
\end{widetext}
The key point is that $DD$ depends on both geometry (through $\lambda$ and $W$) and clustering (through $\xi$).

\paragraph*{Expectation of $DR(b)$.} Because $D$ and $R$ are independent, the cross-sum factorizes at first order:
\begin{widetext}
\begin{align}
\mathbb{E}\!\left[\sum_{i\in D}\sum_{a\in R} W^R_{ia}K_b(r_{ia})\right]
&=
\iint \lambda(x_1)\lambda_R(x_2)\,W^R(x_1,x_2)\,K_b(r_{12})\,dx_1dx_2,
\\
\mathbb{E}\!\left[\sum_{i\in D}\sum_{a\in R} W^R_{ia}\right]
&=
\iint \lambda(x_1)\lambda_R(x_2)\,W^R(x_1,x_2)\,dx_1dx_2.
\end{align}
\end{widetext}
Thus,
\begin{equation}\label{eq:EDR}
\mathbb{E}[DR(b)] = \frac{\iint \lambda_1\lambda_{R,2}\,W^R_{12}\,K_b(r_{12})\,dx_1dx_2}{\iint \lambda_1\lambda_{R,2}\,W^R_{12}\,dx_1dx_2},
\end{equation}
where $\lambda_{R,2}:=\lambda_R(x_2)$.

Importantly, $\mathbb{E}[DR(b)]$ does \emph{not} contain $\xi$: it depends only on the first-order intensities and the geometry/weights.

\paragraph*{Expectation of $RR(b)$.} For the Poisson random catalog $R$,
\begin{widetext}
\begin{align}
\mathbb{E}\!\left[\sum_{a<b\in R} W^R_{ab}K_b(r_{ab})\right] &= \frac12\iint \lambda_R(x_1)\lambda_R(x_2)\,W^R(x_1,x_2)\,K_b(r_{12})\,dx_1dx_2,
\\
\mathbb{E}\!\left[\sum_{a<b\in R} W^R_{ab}\right] &= \frac12\iint \lambda_R(x_1)\lambda_R(x_2)\,W^R(x_1,x_2)\,dx_1dx_2,
\end{align}
\end{widetext}
so
\begin{equation}\label{eq:ERR}
\mathbb{E}[RR(b)] = \frac{\iint \lambda_{R,1}\lambda_{R,2}\,W^R_{12}\,K_b(r_{12})\,dx_1dx_2}
{\iint \lambda_{R,1}\lambda_{R,2}\,W^R_{12}\,dx_1dx_2}.
\end{equation}

Again, $RR$ is purely geometric.

When $\lambda_R(x)$ is proportional to $\lambda(x)$ on the support of the survey,
\begin{equation}\label{eq:random_matches_data}
\lambda_R(x) = c\,\lambda(x)\qquad\text{for some constant }c>0,
\end{equation}
then substituting into Equation~\eqref{eq:EDR} and Equation~\eqref{eq:ERR} shows that the constant cancels and
\begin{equation}\label{eq:DR_equals_RR}
\mathbb{E}[DR(b)] = \mathbb{E}[RR(b)] \;=: q_b.
\end{equation}
We note that $\lambda_R(x) = c\,\lambda(x)$ holds when the survey selection $\varphi(x)$ and random selection $\varphi_R(x)$ are equal. A key implication is that if the data sample has additional spatial modulation beyond the mask (e.g.\ due to varying completeness or target selection), that modulation must be encoded in the randoms as well.

Define $q_b:=\mathbb{E}[RR(b)]$. In the large-random limit, $RR(b)$ concentrates around $q_b$ and similarly $DR(b)$ concentrates around $\mathbb{E}[DR(b)]$.
Using a standard continuous-mapping/Slutsky approximation (ratio of concentrated quantities), we obtain
\begin{equation}\label{eq:ELS_approx}
\mathbb{E}[\widehat\xi_{\rm LS}(b)]
\approx
\frac{\mathbb{E}[DD(b)]-2\,\mathbb{E}[DR(b)]+q_b}{q_b}.
\end{equation}

Under the matching condition Equation~\eqref{eq:random_matches_data}, so that Equation~\eqref{eq:DR_equals_RR} holds, the purely geometric baseline cancels
\begin{equation}\label{eq:ELS_cancel}
\mathbb{E}[\widehat\xi_{\rm LS}(b)]
\approx
\frac{\mathbb{E}[DD(b)]-q_b}{q_b}.
\end{equation}

To interpret the remaining term, consider the regime where $\xi$ is not so large that it materially affects the normalization in Equation~\eqref{eq:EDD_exact_ratio}. This is the usual large-scale regime, and is the approximation implicitly used in most LS derivations (see Appendix~\ref{app:approximation_xi}). Replacing the denominator in Equation~\eqref{eq:EDD_exact_ratio} by its $\xi=0$ counterpart yields the standard first-order form
\begin{widetext}    
\begin{equation}\label{eq:EDD_firstorder}
\mathbb{E}[DD(b)]
\approx
\frac{\iint \lambda_1\lambda_2\,W_{12}\,K_b(r_{12})\,dx_1dx_2}
{\iint \lambda_1\lambda_2\,W_{12}\,dx_1dx_2}
+
\frac{\iint \lambda_1\lambda_2\,\xi(r_{12})\,W_{12}\,K_b(r_{12})\,dx_1dx_2}
{\iint \lambda_1\lambda_2\,W_{12}\,dx_1dx_2}.
\end{equation}
\end{widetext}
The first term is precisely the geometric fraction of pairs in bin $b$ under the (weighted) selection, i.e.\ it matches $q_b$ when $W^R=W$ and $\lambda_R\propto\lambda$.
Hence, under (i) randoms matching the data selection, (ii) consistent pair weighting $W^R=W$, and (iii) the mild first-order normalization approximation, we obtain
\begin{equation}\label{eq:ELS_binavg}
\mathbb{E}[\widehat\xi_{\rm LS}(b)]
\approx
\frac{\iint \lambda_1\lambda_2\,\xi(r_{12})\,W_{12}\,K_b(r_{12})\,dx_1dx_2}
{\iint \lambda_1\lambda_2\,W_{12}\,K_b(r_{12})\,dx_1dx_2}
\;=:\;\overline{\xi}_b.
\end{equation}
That is, LS returns a \emph{pair-weighted bin average} of $\xi$ over the kernel $K_b$ and the effective selection $\lambda$.
When $K_b$ is narrow and $\xi(r)$ varies slowly across the bin, $\overline{\xi}_b\approx \xi(r_b)$ for a representative separation $r_b$.

\section{Exact expression and a controlled approximation} \label{app:approximation_xi}

Starting from Campbell's theorem we obtain the exact expectation of the normalized $DD$ histogram:
\begin{equation}\label{eq:DD_exact}
\mathbb{E}[DD(b)]
=
\frac{A_b + B_b}{C + D},
\end{equation}
where we have defined the four integrals
\begin{align}
A_b &:= \iint \lambda_1\lambda_2\,W_{12}\,K_b(r_{12})\,dx_1dx_2, \label{eq:A_def}\\
B_b &:= \iint \lambda_1\lambda_2\,\xi(r_{12})\,W_{12}\,K_b(r_{12})\,dx_1dx_2, \label{eq:B_def}\\
C   &:= \iint \lambda_1\lambda_2\,W_{12}\,dx_1dx_2, \label{eq:C_def}\\
D   &:= \iint \lambda_1\lambda_2\,\xi(r_{12})\,W_{12}\,dx_1dx_2. \label{eq:D_def}
\end{align}
Here $\lambda_i=\lambda(x_i)$ and $W_{12}=W(x_1,x_2)$ for brevity. Equation~\eqref{eq:DD_exact} is algebraically identical to the split you wrote (Equation~\eqref{eq:EDD_split}) since
\[
\frac{A_b+B_b}{C+D}=\frac{A_b}{C+D}+\frac{B_b}{C+D}.
\]

Write
\begin{equation}\label{eq:ratio_expansion}
\frac{A_b+B_b}{C+D}
=
\frac{A_b}{C} + \frac{B_b}{C} - \frac{A_b\,D}{C^2} + \mathcal{O}\!\bigl((D/C)^2\bigr).
\end{equation}
Thus the leading clustering contribution is $B_b/C$ (the usual numerator-with-$\xi$ divided by the geometric normalization), while the correction term $A_b D/C^2$ is of order $D/C$ relative to the baseline $A_b/C$. If the total clustering contribution to the \emph{denominator} (measured by $D$) is small compared to the geometric denominator $C$, then dropping $D$ is a controlled first-order approximation.

Two important regimes justify $D/C\ll1$:
\begin{itemize}
\item \textbf{Weak clustering on normalization scales.} $\xi(r)$ is appreciable only on small separations and the denominator integrates over pair separations more broadly, so $D$ is typically small compared to $C$. This is often true when using moderately broad pair weights or when the footprint averages many separations.
\item \textbf{Large survey / many pairs.} $C$ scales with the total number of possible pairs (geometry), while $D$ is weighted by $\xi$ and does not grow as fast; in large-area surveys the geometric term dominates.
\end{itemize}

\bibliography{biblio}

\end{document}